\documentclass[aps,pra,10pt,tightenlines,superscriptaddress,amsmath,amssymb,amsfonts,twocolumn,eqsecnum,nofootinbib]{revtex4-1}

\usepackage{times}
\usepackage{amsthm}
\usepackage{graphicx}
\usepackage{color}
\usepackage{bm}
\usepackage{mathtools}
\usepackage{qcircuit}
\usepackage[colorlinks=false, linkcolor=blue, citecolor=magenta]{hyperref}
\usepackage{multirow}
\usepackage{mdwlist}

\usepackage{tikz}

\long\def\beginpgfgraphicnamed#1#2\endpgfgraphicnamed{\includegraphics{#1.pdf}}

\usetikzlibrary{chains,calc,fit,shapes.geometric,decorations.pathmorphing}
\pgfrealjobname{CVgraph_trans}

\theoremstyle{plain}
\newtheorem{theorem}{Theorem}[section]
\newtheorem*{definition}{Definition}

\theoremstyle{remark}

\DeclareMathOperator{\tr}{tr}

\DeclareMathOperator{\Sp}{Sp}
\DeclareMathOperator{\diag}{diag}

\DeclareMathOperator{\cov}{cov}
\DeclareMathOperator{\var}{var}
\renewcommand{\Re}{\operatorname{Re}}
\renewcommand{\Im}{\operatorname{Im}}

\DeclareMathOperator{\sech}{sech}

\def\cc{\text{c.c.}}

\def\tp{\mathrm{T}}
\def\herm{\mathrm{H}}

\def\olcite#1{\cite{#1}}

\def\graphrule#1{\textbf{#1}}

\def\clap#1{\hbox to 0pt{\hss#1\hss}}

\def\mathclap{\mathpalette\mathclapinternal}

\def\mathclapinternal#1#2{\clap{$\mathsurround=0pt#1{#2}$}}

\graphicspath{{graphics/}}

\def\negspace{\!}

\def\rsub#1#2{{#1} \negspace {\protect\vphantom{#1}}_{#2}}

\def\ketsub#1#2{\rsub {\ket{#1}} {#2}}

\def\pket#1{\ketsub{#1} p}

\def\avg#1{\left\langle {#1} \right\rangle}
\def\avgg#1{\langle {#1} \rangle}

\def\micronodesize{4pt}
\def\edgethickness{very thick}
\def\poscolor{blue!60!black}
\def\negcolor{red!30!yellow}

\def\unitcolor{\poscolor}
\def\evencolor{black}
\def\oddcolor{white}
\def\phantomcolor{gray!75}

\tikzset{micro-no-color/.style={
	circle,
	minimum size=\micronodesize,
	inner sep=0pt,
	outer sep=0pt
	}
}

\tikzset{micro/.style={
	micro-no-color, ball color=black,
	}
}

\tikzset{micro-even/.style={
	micro-no-color, ball color=\evencolor,
	}
}

\tikzset{micro-odd/.style={
	micro-no-color, ball color=\oddcolor,
	}
}

\tikzset{poslink/.style={
	\edgethickness,
	draw=\poscolor,
	}
}

\tikzset{neglink/.style={
	\edgethickness,
	draw=\negcolor,
	}
}

\tikzset{unitlink/.style={
	\edgethickness,
	draw=\unitcolor,
	}
}

\tikzset{phantomlink/.style={
	\edgethickness,
	draw=\phantomcolor,
	ultra nearly transparent
	}
}

\tikzset{nodehighlight/.style={
	thin,
	red,
	fill=red,
	semitransparent
	}
}

\tikzset{latticeopts/.style={
	x=3.375cm,
	y=3.75mm,
	z=20.25mm,
	inner sep=0pt,
	outer sep=0pt
	}
}

\def\reals{\mathbb{R}}

\def\op#1{\hat{#1}}
\def\opvec#1{\op{\vec{#1}}}
\def\opmat#1{\op{\mat{#1}}}
\def\id{I}

\def\1{\mat{\id}}

\def\mat#1{\bm{\mathrm{#1}}}

\renewcommand{\vec}[1]{\bm{\mathrm{#1}}}

\def\controlled#1{\mathrm{C}_{#1}}
\def\CZ{\controlled Z}

\def\opCZ{\op{\mathrm{C}}_{Z}}

\def\cH{\mathcal{H}}

\begin{document}

\title{Graphical calculus for Gaussian pure states}

\author{Nicolas C. Menicucci}
\author{Steven T. Flammia}
\affiliation{Perimeter Institute for Theoretical Physics, Waterloo, Ontario N2L 2Y5, Canada}

\author{Peter van Loock}
\affiliation{Optical Quantum Information Theory Group,
Max Planck Institute for the Science of Light}
\affiliation{Institute of Theoretical Physics I, Universit\"{a}t Erlangen-N\"{u}rnberg,
Staudtstr.7/B2, 91058 Erlangen, Germany}

\date{February 4, 2011}

\begin{abstract} 
We provide a unified graphical calculus for all Gaussian pure states, including graph transformation rules for all local and semi-local Gaussian unitary operations, as well as local quadrature measurements. We then use this graphical calculus to analyze continuous-variable~(CV) cluster states, the essential resource for one-way quantum computing with CV systems.  Current graphical approaches to CV cluster states are only valid in the unphysical limit of infinite squeezing, and the associated graph transformation rules only apply when the initial and final states are of this form.  Our formalism applies to all Gaussian pure states and subsumes these rules in a natural way.  In addition, the term ``CV graph state'' currently has several inequivalent definitions in use.  Using this formalism we provide a single unifying definition that encompasses all of them.  We provide many examples of how the formalism may be used in the context of CV cluster states: defining the ``closest'' CV cluster state to a given Gaussian pure state and quantifying the error in the approximation due to finite squeezing; analyzing the optimality of certain methods of generating CV cluster states; drawing connections between this new graphical formalism and bosonic Hamiltonians with Gaussian ground states, including those useful for CV one-way quantum computing; and deriving a graphical measure of bipartite entanglement for certain classes of CV cluster states.  We mention other possible applications of this formalism and conclude with a brief note on fault tolerance in CV one-way quantum computing.
\end{abstract}

\maketitle

\section{Introduction}
\label{sec:intro}

The invention of one-way quantum computing (QC)~\cite{Raussendorf2001} launched an intensive research effort into this new method of QC that eliminates unitary evolution and relies solely on adaptive measurements on a highly entangled state of many qubits called a \emph{cluster state}~\cite{Briegel2001}.  Concurrently, other work was underway generalizing qubit-based QC to QC using continuous-variable~(CV) systems~\cite{Braunstein2005a}.  The two paths merged with the invention of CV cluster states~\cite{Zhang2006}, which were shortly thereafter shown to be capable of serving as the entangled resource in the CV~version of one-way~QC~\cite{Menicucci2006,Gu2009}.

Initially, CV cluster states and the platform of one-way QC making use of them were not believed to be a promising contender for scalable QC~\cite{Menicucci2006}.  It was, however, believed that CV cluster states would be convenient for demonstrating the basic principles of one-way QC since generating such states in the optical context was easier than making ordinary cluster states from optical qubits~\cite{Nielsen2004, Browne2005}. The main reason for this belief was that CV cluster states could be generated deterministically, while getting single photons to interact required nondeterministic gates whose (heralded) failure happens a large fraction of the time~\cite{Knill2001}.  Nevertheless, \emph{ideal} CV cluster states are not achievable since they are singular states (i.e.,~infinitely squeezed) and thus have an infinite average photon number and infinite energy.  Approximate states must therefore be used instead, necessarily leading to errors in any CV one-way QC protocol~\cite{Gu2009,Ohliger2010} (we will say more about fault tolerance in Section~\ref{sec:conc}).  The most natural choice for these approximate states would be multimode squeezed states, but the originally proposed method of making them~\cite{Menicucci2006} involved experimentally arduous inline squeezing operations~\cite{Yurke1985, LaPorta1989, Yoshikawa2007}.  This limited the expected usefulness of the technology.

Shortly after the invention of CV one-way QC, it was shown that inline squeezing was not required at all and that CV cluster states could be generated optically using offline squeezing plus interferometry~\cite{vanLoock2007}. This method involves preparing one single-mode squeezed vacuum state per node of the cluster and sending these states through an appropriately designed network of beamsplitters.  In fact, this method can be used to make any Gaussian state at all~\cite{Braunstein2005}.  This represented a vast simplification for experiments, which quickly demonstrated the viability of this method of generating CV cluster states and their usefulness for simple CV quantum information processing tasks~\cite{Su2007, Yukawa2008, Yonezawa2010, Ukai2010, Miwa2010}.

Concurrent with this work was a separate initiative to generate optical CV cluster states in a single-shot, top-down fashion using just a single optical parametric oscillator~(OPO) consisting of a nonlinear crystal within an optical cavity~\cite{Menicucci2007}.  In this implementation, independent modes are not separated in space (as in previous optical proposals) but are instead taken to be the different frequencies within a single beam.  The initial proposal showed that a single OPO and appropriately designed multifrequency pump beam could, in principle, generate any approximate CV cluster state with a bipartite graph,\footnote{The reader should be aware that we use the term ``cluster state'' where other authors might prefer ``graph state,'' since sometimes in the literature ``cluster state'' is used to refer only to a graph state with a square-lattice graph.  We would refer to such states as ``cluster states with a square-lattice graph.''  Each convention has its proponents, but in the present context, where ``CV graph state'' could have three different meanings (CV cluster state, $\cH$-graph state, or general Gaussian pure state), this convention also serves a clarifying function.}  which includes the universal family of square-lattice graphs.  Further work revealed that this method could generate a multitude of small CV cluster states~\cite{Zaidi2008} or a universal family of CV cluster states~\cite{Menicucci2008, Flammia2009}, using a method that has excellent scaling potential up to a few thousand optical modes with currently available technology~\cite{Pooser2005,Pysher2010,Pysher2009a,Midgley2010}.

Yet another method~\cite{Menicucci2010} reintroduces the experimentally challenging $\CZ$~gate.  But in this case, only \emph{one} such gate is needed because the modes are encoded \emph{temporally}, each traversing the same optical path (but at different times) and each passing multiple times through the same optical hardware implementing the $\CZ$~gate.  This method has the additional advantage that the cluster state is extended as needed---simultaneously with measurements implementing an algorithm on it---in a manner analogous to repeatedly laying down additional track in front of a moving train car (a ``Wallace and Gromit'' approach; see footnote in Ref.~\olcite{Menicucci2010}).  Such a method eliminates the need for long-time coherence of a large cluster state because only a small piece of the state exists at any given time.

While every CV cluster state---regardless of how it is made---can be represented by a graph~\cite{Gu2009}, the single-OPO generation method revealed another type of graph that is useful for describing Gaussian pure states~\cite{Pfister2007}.  This graph indicates the strength and pairings of the two-mode squeezing interactions that act within the OPO, and its adjacency matrix defines the interaction Hamiltonian directly.  Thus, we call these graphs Hamiltonian graphs~\cite{Menicucci2007,Zaidi2008, Menicucci2008, Flammia2009,Flammia2009a}, or $\cH$-graphs for short.\footnote{In Ref.~\olcite{Menicucci2007} the term ``two-mode-squeezing graph'' was used instead of ``$\cH$-graph.''  These terms are synonymous, and only the latter will be used in this paper.}  Despite this natural way of representing the Hamiltonian interactions by graphs, the resulting states are \emph{not} ``CV graph states'' in the sense of a CV cluster state with the same graph as the $\cH$-graph, although they can be interpreted as CV cluster states with (in general) a \emph{different} graph~\cite{Menicucci2007}.  This creates an ambiguity in the meaning of a ``CV graph state.''

Independently of this work, Zhang showed that ideal CV cluster states admit graph transformation rules that correspond to local Gaussian operations mapping them to other ideal CV cluster states~\cite{Zhang2008}.  These rules bear some similarity to the corresponding rules for qubit cluster states~\cite{VandenNest2004, Hein2004}
but they are not exactly equivalent.  (Related work has also been done for qubit stabilizer states and local Clifford transformations~\cite{Schlingemann2002,Schlingemann2004,Elliott2008,Elliott2010}.)  Further revisions showed that ideal CV cluster states with \emph{weighted} graphs were necessary for a more complete understanding of the graph transformation rules~\cite{Zhang2008a}---a consequence of the continuity of the quantum variables in question, as opposed to the binary nature of qubits, whose graphs are necessarily unweighted.\footnote{Weighted graph states for qubits have been defined~\cite{Hein2006} but they are not stabilizer states.  When dealing with CVs, however, weighted graphs occur naturally because the entangling operation that makes a CV cluster state necessarily has a strength (which can be---but need not be---chosen equal for all interactions).  This strength becomes the weight for the corresponding edge in the graph.  Unlike their similarly named but non-analogous qubit counterparts, CV cluster states with weighted graphs \emph{are} CV stabilizer states.}  The effect of quadrature measurements, which can be used to implement any Gaussian operation in CV one-way QC~\cite{Ukai2009,Gu2009}, has recently been incorporated into the formalism, as well~\cite{Zhang2010}.

This original graphical calculus, while useful for demonstrating local Gaussian equivalence of CV cluster states, has several limitations.  First, for the weighted as well as for the unweighted case, only ideal (i.e.,~infinitely squeezed) CV cluster states can be represented.  As mentioned before, these states are not physical.  Neither their approximating Gaussian states nor any other Gaussian state can be represented in the formalism.  Second, there are many Gaussian operations (for instance, the very common Fourier transform) that do not map CV cluster states to other CV cluster states and thus cannot be represented as a graph transformation.  Third, no connection is made with $\cH$-graphs; the rules apply only to CV cluster-state graphs.  Nonetheless, Zhang's formalism is exciting because it promises an intuitive visual way of manipulating CV cluster states, paralleling similar tools for qubit cluster states~\cite{VandenNest2004, Hein2004}.

In this paper we generalize these rules in a consistent fashion to cover all Gaussian pure states, including approximate CV cluster states.  This includes physical states generated by the action of a Hamiltonian with an associated $\cH$-graph.  Because the details of a Gaussian pure state are displayed in its graphical representation (and represented \emph{uniquely} within it), this formalism can be used to quantify the deviations from ideality for any approximate CV cluster state.  Furthermore, the formalism can also be used to identify the ``closest'' CV cluster state to any given Gaussian pure state, and it is useful when considering physical systems whose ground states would be useful for CV one-way QC~\cite{Doherty2009}.  We also make connections with a measure of bipartite entanglement in Gaussian pure states. In certain cases, this admits a simple graphical rule.

In what follows, we shall (1)~define the unique graph associated with any Gaussian pure state, (2)~derive the transformation rules for all Gaussian unitary operations, (3)~illustrate the effect of local Gaussian unitary operations in graphical form, showing that they faithfully generalize Zhang's rules~\cite{Zhang2008a,Zhang2010}, and (4)~illustrate the connection of this formalism to approximate CV cluster state generation via an $\cH$-graph Hamiltonian, and (5)~provide several applications of the formalism to analysis of physical states and the $\cH$-graph generation method, as well as Hamiltonian ground states and bipartite entanglement.  The connection with $\cH$-graphs answers an important question about the method proposed in Refs.~\cite{Menicucci2007,Zaidi2008,Menicucci2008,Flammia2009}---namely, what happens when the method is used on \emph{physical} states.  Previous connections between $\cH$-graphs and CV cluster-state graphs have only been rigorously made in the unphysical limit of infinite squeezing.  This formalism allows the important effects of finite squeezing to be properly accounted for while remaining entirely within the intuitive framework of the graphical representation.

The mathematics behind this formalism is the complex matrix formalism for representing and manipulating Gaussian pure states~\cite{Simon1988}.  When we interpret these matrices as adjacency matrices for complex-weighted graphs, transformations using the symplectic representation can also be interpreted in graph-theoretic terms.  In this formalism, real-weighted graphs, representing idealized, unphysical, infinitely squeezed Gaussian states~\cite{Zhang2008,Zhang2008a,Zhang2010}, are generalized to complex-weighted graphs that uniquely specify realistic, physical, finitely squeezed Gaussian pure states.  This extension includes generalizing the real-valued graph-state nullifiers for ideal CV cluster states~\cite{Gu2009} to complex-valued nullifiers for physical CV cluster states. These results are closely related to the stabilizer formalism for Gaussian pure states, which, though utilized to some extent for a proof of the CV version of the Gottesman-Knill theorem~\cite{Bartlett2002} and often mentioned as a straightforward generalization from finite-dimensional systems~\olcite{Fattal2004}, has not been fully explored yet.  While we are building on existing mathematics, what is new in this work---beyond the straightforward graph-theoretic interpretation of that mathematics---is showing its natural connection to CV cluster states, plus all of the examples, applications, and avenues for future work that open up as a result.  We will have more to say about the context and importance of our work in Section~\ref{sec:conc}.

\section{Gaussian Pure States and Complex-Weighted Graphs}
\label{sec:gauss}

\subsection{CV Cluster States}
\label{subsec:gauss:cvcs}

The motivation for a graphical study of Gaussian pure states begins with CV cluster states~\cite{Zhang2006,Menicucci2006,Gu2009}.  In the ideal case, CV cluster states are prepared beginning with a collection of $N$~zero-momentum eigenstates, which we write as~$\pket{0}^{\otimes N}$, where the $p$-subscripted kets satisfy~$\op p\pket s = s \pket s$.  These states are then entangled via a collection of controlled-$Z$ operations, denoted~$\opCZ = \exp (i g \op q \otimes \op q)$, where the real number~$g$ is the strength of the interaction.  Since all $\CZ$ operations commute, they can be performed in any order (or simultaneously), which leads naturally to the use of graphs as recipes for generating particular CV cluster states.  An example of such a graph---and the CV cluster-state recipe it encodes---is illustrated in Figure~\ref{fig:CVgraph}.  Each node represents a zero-momentum eigenstate, and edges indicate a $\CZ$ operation to be performed between the two connected nodes.  The strength~$g$ of the interaction is indicated by the label, or \emph{weight}, of the associated edge.  As such, \emph{weighted graphs} with real-valued weights are the natural language for depicting ideal CV cluster states.

\begin{figure}[!tbp]
\begin{center}
\beginpgfgraphicnamed{CVCS_orig}
\begin{tikzpicture} [scale=1,label distance=-2pt]

	\def\orig{(0,0)}
	\def\symb{(0.6*\scale,-0.5*\scale)}
	\def\circuitloc{(2.5*\scale,-0.5*\scale)}
	
	\def\scale{1.5}
	
	\def\circuitinnersep{6pt}

	\tiny

	\path \orig [outer sep=0pt]
		node [micro-odd,label=45:1] (1) {}
		++(180:\scale) node [micro-odd,label=135:2] (2) {}
		++(-90:\scale) node [micro-odd,label=225:3] (3) {}
		++(0:\scale) node [micro-odd,label=315:4] (4) {};

	\footnotesize

%
%
	\path (1)
		edge [unitlink] node [above=-1pt] {$1$} (2)
		edge [unitlink] node [below right=-2.5pt] {$3$} (3);

	\path (2)
		edge [unitlink] node [left=-2pt] {$2$} (3);
	
	\path (3)
		edge [unitlink] node [below=-1pt] {$-1$} (4);
	
	\path \symb
		node {$\Longleftrightarrow$};
	
	\small

	\path \circuitloc
	node [inner sep=\circuitinnersep] (circuit)
	{$\mbox{
	\Qcircuit @C=1.80em @R=2.2em @!R @!C {
		\lstick {\ketsub {0}{p_1}}	& \ctrl{1}	&  \qw 		& \ctrl{2}		& \qw	\\
		\lstick {\ketsub {0}{p_2}}	& \control \qw	& \ctrl{1}		& \qw		& \qw	\\
		\lstick {\ketsub {0}{p_3}}	& \ctrl{1}		& \control \qw	& \control \qw	& \qw	\\
		\lstick {\ketsub {0}{p_4}}	& \control \qw	& \qw		& \qw		& \qw	
	}
	}
	$};
	
	\footnotesize
	
	\path (circuit.west)	++(\circuitinnersep-9pt,0)
					+(6em,0) node {$2$}
					+(4em,2.4em) node {$1$}
					+(4em,-2.4em) node [xshift=0.4em] {$-1$}
					+(8em,2.4em) node {$3$};

\end{tikzpicture}
\endpgfgraphicnamed
\caption{Original formulation of weighted graphs for ideal CV cluster states.  Ideal CV cluster states are represented by undirected graphs with real-weighted edges~\cite{Zhang2006,Menicucci2006,vanLoock2007,Menicucci2007,Zaidi2008,Menicucci2008,Flammia2009,Gu2009}.  (Unweighted graphs are a special case with all weights equal to~1.)  Each graph uniquely defines a recipe (i.e.,~a quantum circuit) for creating a CV cluster state, as illustrated above: (1)~each node represents a state that is infinitely squeezed in the $\op p$~quadrature~$\pket 0$; (2)~$\CZ$~gates are applied between modes in accordance with the graph, with the weight~$g$ of an edge corresponding to the strength of the interaction~$\opCZ(g) = e^{ig\op q \otimes \op q}$ between the two nodes connected.  These states are unphysical because they cannot be normalized.  Instead they are approximated in physical applications by very highly squeezed states.}
\label{fig:CVgraph}
\end{center}
\end{figure}
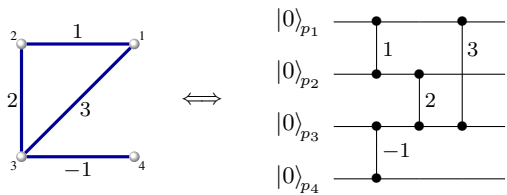

Labeling the nodes of the graph in some arbitrary order, we can define a symmetric \emph{adjacency matrix}~$\mat A=\mat A^\tp$ whose $(j,k)^{\text{th}}$ entry~$A_{jk}$~is equal to the weight of the edge linking node~$j$ to node~$k$ (with no edge corresponding to a weight of~0).  Since a graph is uniquely specified (up to isomorphism) by its adjacency matrix, we will often omit the distinction between the two and refer unambiguously to ``the graph~$\mat A$''.  The collection of controlled-$Z$ operations used to make the CV cluster state is now a function of~$\mat A$, denoted~$\opCZ[\mat A]$.  The CV cluster state associated with the graph~$\mat A$ is then
\begin{align}
\label{eq:psiA}
	\ket{\psi_{\mat A}} &= \opCZ[\mat A] \pket 0^{\otimes N} \nonumber \\
	&= \prod_{\mathclap{j,k=1}}^N \exp \left(\frac i 2 A_{jk} \op q_j \op q_k \right) \pket 0^{\otimes N} \nonumber \\
	&= \exp \left( \frac i 2 \opvec q^\tp \mat A \opvec q \right) \pket {0}^{\otimes N}\,,
\end{align}
where $\opvec q = (\op q_1, \dotsc, \op q_N)^\tp$~is a column vector of Schr\"odinger-picture position operators.  The factor of~2 is necessary because each edge weight appears twice in the sum (as~$A_{jk}$ and as~$A_{kj}$).

Ideal CV cluster states satisfy a set of \emph{nullifier} relations~\cite{Gu2009,Menicucci2007}, which can be written concisely as
\begin{align}
\label{eq:psiAnull}
	(\opvec p - \mat A \opvec q) \ket {\psi_{\mat A}} = \vec 0\,,
\end{align}
where $\opvec p = (\op p_1, \dotsc, \op p_N)^\tp$~is a column vector of Schr\"odinger-picture momentum operators.  This equation actually represents $N$~independent equations, one for each component of the vector~$(\opvec p - \mat A \opvec q)$, which are called \emph{nullifiers} for the state~$ \ket {\psi_{\mat A}}$ because that state is a simultaneous zero-eigenstate of them (and of any linear combination of them).

These ideal CV cluster states admit a convenient graphical representation in terms of the adjacency matrix~$\mat A$.  Some local Gaussian unitary operations~\cite{Zhang2008a} and quadrature measurements~\cite{Zhang2010} can be represented as convenient graphical update rules.  As discussed in the introduction, this graphical formalism is elegant and intuitive but also has several limitations---most notably, the restriction to ideal (and hence non-physical) CV cluster states and to a subset of all local Gaussian unitaries (those which map ideal CV cluster states to other ideal CV cluster states).

Approximate CV cluster states are those for which measurements of each nullifier give values that are \emph{close to zero}.  Quantifying this, we would say that an approximate CV cluster state is any member of a family of Gaussian pure states, indexed by an overall squeezing parameter~$\alpha$ for which
\begin{align}
\label{eq:approxCVCS}
	\lim_{\alpha \to \infty} \cov{(\opvec p - \mat A \opvec q)} = \mat 0\,,
\end{align}
where the covariance matrix of a vector of operators has components defined as the symmetrized expectation value
\begin{align}
\label{eq:covcomponents}
	(\cov \opvec r)_{jk} := \frac 1 2 \avg{ \{\op r_j^\dag, \op r_k\} }\,,
\end{align}
with the expectation taken over the state of interest---in this case, the $\alpha$-indexed approximate CV cluster state~$\ket {\psi_{\mat A}(\alpha)}$.  (Notice that this definition assumes that the state has zero mean.)  The limit in Eq.~\eqref{eq:approxCVCS} is defined component-wise for all entries in the covariance matrix, but because the matrix is positive definite, it is sufficient to require that the relation hold only for the diagonal elements:
\begin{align}
\label{eq:approxCVCSexplicit}
	\lim_{\alpha \to \infty} \bra {\psi_{\mat A}(\alpha)} \bigl( \op p_j - \sum_k A_{jk} \op q_k  \bigr)^2 \ket {\psi_{\mat A}(\alpha)} = 0 \quad \forall j\,.
\end{align}
Any $\alpha$-indexed family of Gaussian pure states~$\{\ket {\psi_{\mat A}(\alpha)}\}$ satisfying Eq.~\eqref{eq:approxCVCS} [or Eq.~\eqref{eq:approxCVCSexplicit}] defines a family of approximate CV cluster states with graph~$\mat A$.  We would like these Gaussian pure states to be representable directly in a graphical formalism in a manner that takes into account their unique deviation from ideality.

In addition to this, there are plenty of other useful Gaussian pure states that are not approximate CV cluster states at all.  For instance, a two-mode squeezed state satisfies
\begin{align}
\label{eq:TMSstate}
	\begin{split}
		\var(\op q_1 - \op q_2) &= e^{-2\alpha}\,, \\
		\var(\op p_1 + \op p_2) &= e^{-2\alpha}\,.
	\end{split}
\end{align}
Such states are readily made in the lab by passing position-squeezed and momentum-squeezed beams through a 50:50 beamsplitter~\cite{Furusawa1998,vanLoock1999,Braunstein2005a} or directly by nondegenerate parametric downconversion~\cite{Reid1988,Reid1989,Drummond1990,Ou1992a,Braunstein2005a}.  Despite being an $\alpha$-indexed family of multimode squeezed states whose variances tend to zero as $\alpha \to \infty$, these states do not satisfy Eq.~\eqref{eq:approxCVCS} for any choice of (finitely-weighted) graph~$\mat A$, which means they cannot be represented within the existing graphical formalism for CV cluster states~\cite{Zhang2008a}---even in the limiting case where $\alpha \to \infty$.  This is unfortunate since the two-mode squeezed state is related to a two-mode CV cluster state by a Fourier transform on one of the modes---a local Gaussian unitary~\cite{Menicucci2007} that is one of the simplest to perform experimentally.  That this equivalence cannot be represented in the graphical formalism---in either the ideal or the approximate case---is a significant drawback.  Our formalism addresses all these concerns.

\subsection{Desired properties of the graphical calculus}
\label{subsec:gauss:props}

In our attempt to generalize the formalism of References~\olcite{Zhang2008a,Zhang2010}, we desire a unified graphical calculus that has the following properties:
\begin{enumerate}
\item All Gaussian pure states can be represented uniquely, up to phase-space displacements, as graphs.
\item All Gaussian unitaries---local or otherwise---can be represented uniquely, up to phase-space displacements, as graph transformations.
\item All local projective measurements of quadrature operators can be represented uniquely, up to phase-space displacements, as graph transformations.
\item The representation of a family of approximate CV cluster states faithfully limits to the standard graph representation of the associated ideal CV cluster state in the limit of large squeezing.
\item The representation of local Gaussian unitaries and projective measurements of quadrature operators acting on a family of approximate CV cluster states faithfully reproduces Zhang's rules~\cite{Zhang2008a,Zhang2010} in the limit of large squeezing.
\end{enumerate}
In addition to these requirements, we would also like the graphical calculus to be useful for the following purposes:
\begin{itemize}
\item Visualize the entanglement structure of a Gaussian pure state.
\item Consider finite squeezing effects within the graphical formalism, including their effect on one-way QC using an approximate CV cluster state.
\item Make a connection with $\cH$-graphs and their usefulness in generating CV cluster states~
\cite{Menicucci2007,Zaidi2008,Menicucci2008}, including possible graph transformation rules between the two types of graphs.
\end{itemize}
This is not meant to be an exhaustive list.  We expect that other uses will present themselves as the formalism gets applied to actual calculations.

\subsection{Symplectic representation of Gaussian pure states}
\label{subsec:gauss:symp}

In the present and the following subsection, we shall review the complex-matrix formalism~\cite{Simon1988} for representing Gaussian pure states and their transformations among each other, adapt it to our notation, and recast it for our purposes.

It is well known that all $N$-mode Gaussian pure states can be created by acting on the ground state of $N$~harmonic oscillator with a unitary operation whose Heisenberg action on the vector of quadrature operators is a symplectic transformation, followed by a phase-space displacement.  These are sometimes called linear unitary Bogoliubov~(LUBO) transformations~\cite{Braunstein2005a}, but we will not use this term.  Furthermore, we will neglect the phase-space displacement altogether since we only desire that our graphical formalism describe the noise properties of the state, which do not depend on overall displacement.

Stacking~$\opvec q$ on top of~$\opvec p$ to form a column vector called $\opvec x = \left( \begin{smallmatrix} \opvec q \\ \opvec p \end{smallmatrix} \right)$, the Heisenberg action of a Gaussian unitary operation~$\op U$ takes the form
\begin{align}
	\opvec x' = \op U^\dag \opvec x \op U = \mat S \opvec x\,,
\end{align}
where $\mat S$ is a symplectic matrix of c-numbers that acts via matrix multiplication on $\opvec x$ as a vector, while $\op U$ is a unitary operator that acts on the individual operators within~$\opvec x$.  Specifically,
\begin{align}\label{eq:specific_symplectic_trafo}
	\op x_j' = \op U^\dag \op x_j \op U = \sum_{k=1}^{2N} S_{jk} \op x_k\,.
\end{align}
Notice that in general there would be a phase-space displacement term, which would give $\opvec x' = \mat S \opvec x + \vec y$, but we are neglecting this.  There is a unique~$\mat S$ for every Gaussian~$\op U$, and there is a unique~$\op U$ (up to an overall phase) for every symplectic~$\mat S$.  This correspondence can be chosen to faithfully preserve composition and map the identity operator~$\op \id$ to the identity matrix~$\mat \id$, thus giving a symplectic representation of the Gaussian unitary group~\cite{Arvind1995}.

The symplectic nature of~$\mat S$ is guaranteed because the commutation relations must be preserved, giving rise to a symplectic form~$\mat \Omega$ to be preserved by the Heisenberg matrix action.  The explicit form of~$\mat \Omega$ may be deduced by writing out the commutation relations for $\opvec x$ and requiring them to be unchanged under the Gaussian unitary operation.  The canonical commutation relations~$[\op q_j, \op p_k] = i \delta_{jk}$ (with $\hbar = 1$) can be written succinctly as
\begin{align}
\label{eq:xcomm}
	[ \opvec x, \opvec x^\tp ] = i
	\begin{pmatrix}
		\mat 0	& \mat \id	\\
		-\mat \id	& \mat 0
	\end{pmatrix}
	=: i \mat \Omega\,,
\end{align}
where the commutator of two operator-valued vectors is defined as
\begin{align}
\label{eq:commdef}
	[ \opvec r, \opvec s^\tp ] := \opvec r \opvec s^\tp - (\opvec s \opvec r^\tp)^\tp\,.
\end{align}
Note that the transpose operation acts only on the entries in the matrix (or vector), leaving the actual operators involved alone.%
\footnote{\emph{Note on notational conventions.} %
Because we are dealing with operator-valued matrices (or vectors in this case), the transpose operation must be carefully defined.  In fact, we define it in a natural way to simply exchange the indices of an operator-valued matrix [$(\opmat A{}^\tp)_{jk} = (\opmat A)_{kj}$] and leave the entries themselves alone.  That is, each entry in the matrix---which is itself an operator---does \emph{not} get a transpose applied to it.  It is then no longer the case that matrix transposition follows the usual distributive rule---i.e.,\ $(\opmat A \opmat B)^\tp \neq \opmat B{}^\tp \opmat A{}^\tp$ because the operators end up in the wrong order.  Rather than being a problem, we can use this feature to define the commutator-product of two operator-valued column vectors~$\opvec r$ and~$\opvec s$ as in Eq.~\eqref{eq:commdef}.  This has the desired property of forming a matrix whose entries are the commutators in question:
\begin{align*}
	\left ([ \opvec r, \opvec s^\tp ] \right)_{jk} = \op r_j \op s_k - \op s_k \op r_j = [ \op r_j, \op s_k ]\,.
\end{align*}
This is how Eq.~\eqref{eq:xcomm} should be interpreted.  Also note that c-number matrices acting on the vectors within the commutator factor out:
\begin{align*}
	[ \mat A \opvec r, (\mat B \opvec s)^\tp ] &= \mat A \opvec r \opvec s^\tp \mat B^\tp - (\mat B \opvec s \opvec r^\tp \mat A^\tp)^\tp \nonumber \\
	&= \mat A \opvec r \opvec s^\tp \mat B^\tp - \mat A (\opvec s \opvec r^\tp)^\tp \mat B^\tp \nonumber \\
	&= \mat A [\opvec r, \opvec s^\tp ] \mat B^\tp\,.
\end{align*}
The usefulness here stems from the ability to ``vectorize'' expressions using operators in a natural way.  This will be useful in what follows.}  %
Requiring that the quadrature-operator commutators remain unchanged after the Gaussian operation gives
\begin{align}
	i \mat \Omega &=[ \opvec x', \opvec x'^\tp] \nonumber \\
	&=[ \mat S \opvec x, (\mat S\opvec x)^\tp] \nonumber \\
	&= \mat S[ \opvec x, \opvec x^\tp] \mat S^\tp \nonumber \\
	&= i \mat S \mat \Omega \mat S^\tp\,.
\end{align}
Noting that the $i$'s cancel, this equation is exactly the defining relation for any ${2N\times 2N}$ square matrix $\mat S$~to be a symplectic matrix with symplectic form~$\mat \Omega$.  The symplectic nature of~$\mat S$ is therefore guaranteed and required by the need to preserve the canonical commutation relations, which themselves play the role of the symplectic form~$\mat \Omega$ (up to an overall factor of~$i$), as shown in Eq.~\eqref{eq:xcomm}.

A Gaussian pure state (with zero mean) is uniquely specified by its covariance matrix.  We will write the symmetrized covariance matrix for an operator-valued vector as
\begin{align}
\label{eq:covmat}
	\cov \opvec r &= \frac 1 2 \avg{ \{ \opvec r^\dag, \opvec r^\tp \} }\,,
\end{align}
which accords with Eq.~\eqref{eq:covcomponents} if we define the anti-commutator product as
\begin{align}
\label{eq:anticommdef}
	\{ \opvec r, \opvec s^\tp \} := \opvec r \opvec s^\tp + (\opvec s \opvec r^\tp)^\tp\,,
\end{align}
which is analogous to Eq.~\eqref{eq:commdef}, and if we require that Hermitian conjugation (indicated by~$^\dag$) apply only to the operators within the vector \emph{without} transposing the vector itself.  We will use the notation $^\herm$ to indicate transposition of the vector and conjugation of its entries:
\begin{align}
\label{eq:hermdef}
	\opvec r^\herm := (\opvec r^\dag)^\tp = (\opvec r^\tp)^\dag\,.
\end{align}
These caveats are unimportant for our current purposes because~$\opvec x = \opvec x^\dag$, but they will be necessary later on when we wish to take the covariance matrix associated with non-Hermitian operators.  Even in those cases, Eqs.~\eqref{eq:covcomponents} and~\eqref{eq:covmat} still hold.

Since every $N$-mode Gaussian pure state can be obtained by acting with a Gaussian unitary operation on the ground state of $N$~independent harmonic oscillators, we can use the symplectic representation of this operation to parametrize these states.  To eliminate units in~$\opvec q$ and~$\opvec p$, we will normalize the covariance matrix of the $N$-mode ground state to be
\begin{align}
\label{eq:covground}
	\cov \opvec x_0 = \frac 1 2 \mat \id\,,
\end{align}
where $\opvec x_0$~is the vector of Heisenberg operators associated with this state.  This means that $\var{\op q_{0j}} = \var{\op p_{0j}} = \tfrac 1 2$ for every mode~$j$.  The Heisenberg operators for any Gaussian pure state can be obtained from $\opvec x_0$ by acting with a symplectic matrix, resulting in a covariance matrix of
\begin{align}
\label{eq:covfromS}
	\cov \opvec x &= \cov (\mat S \opvec x_0) \nonumber \\
	&= \frac 1 2 \avg{ \{ (\mat S \opvec x_0)^\dag, (\mat S \opvec x_0)^\tp \} } \nonumber \\
	&= \frac 1 2 \mat S \avg{ \{ \opvec x_0^\dag, \opvec x_0^\tp \} } \mat S^\tp \nonumber \\
	&= \mat S (\cov \opvec x_0) \mat S^\tp \nonumber \\
	&= \frac 1 2 \mat S \mat S^\tp\,.
\end{align}
(For some thoughts on how to generalize this to mixed Gaussian states, see Appendix~\ref{app:mixed}.)  Since the covariance matrix uniquely defines a Gaussian state, by Eq.~\eqref{eq:covfromS} so does~$\mat S \mat S^\tp$.  To be practically useful for our purposes, we need a graph representation of~$\mat S$ (or more accurately, of~$\mat S \mat S^\tp$) and useful transformation rules for this representation.  To this end, we will decompose~$\mat S$ and use the resulting matrix factors to define the adjacency matrix for an associated graph.

There are a number of ways to decompose a symplectic matrix, but the one we will be interested in is the following \emph{unique} decomposition for any symplectic~$\mat S$~\cite{Simon1988}:
\begin{align}
\label{eq:preiwasawa}
	\mat S =
	\begin{pmatrix}
		\mat \id	& \mat 0	\\
		\mat V	& \mat \id
	\end{pmatrix}
	\begin{pmatrix}
		\mat U^{-1/2}	& \mat 0	\\
		\mat 0		& \mat U^{1/2}
	\end{pmatrix}
	\begin{pmatrix}
		\mat X	& -\mat Y	\\
		\mat Y	& \mat X
	\end{pmatrix}
	\,,
\end{align}
where $\mat U$~is symmetric and positive definite ($\mat U = \mat U^\tp >~0$), $\mat V$~is symmetric (but not necessarily positive definite), and the third matrix is orthogonal and thus irrelevant in the product~$\mat S \mat S^\tp$.\footnote{Physically, in an optical setting for instance, this term corresponds to a passive interferometer, which can be implemented using just beamsplitters and phase shifters.  These have no effect on the vacuum.}  Therefore, while this expansion is unique for a given~$\mat S$~\cite{Simon1988}, since we only care about~$\mat S \mat S^\tp$ we can fix~$\mat X = \mat \id$ and~$\mat Y = \mat 0$ and, after multiplying the right hand side above, \emph{define}
\begin{align}
\label{eq:Suv}
	\mat S_{(\mat U, \mat V)} :=
	\begin{pmatrix}
		\mat U^{-1/2}			& \mat 0	\\
		\mat V \mat U^{-1/2}		& \mat U^{1/2}
	\end{pmatrix}
	\,.
\end{align}
Using Eq.~\eqref{eq:covfromS}, the covariance matrix associated with this state is
\begin{align}
\label{eq:covmatuv}
	\mat \Sigma_{(\mat U, \mat V)} &= \frac 1 2 \mat S_{(\mat U, \mat V)} \mat S_{(\mat U, \mat V)}^\tp \nonumber \\
	&= \frac 1 2
	\begin{pmatrix}
		\mat U^{-1}			& \mat U^{-1} \mat V	\\
		\mat V \mat U^{-1}	& \mat U + \mat V \mat U^{-1} \mat V
	\end{pmatrix}
	\,.
\end{align}

Using this we can immediately write down the Wigner function for the state
\begin{align}
\label{eq:Wignerfromuv}
	W_{(\mat U, \mat V)}(\vec x) &= (2\pi)^{-N} (\det \mat \Sigma_{(\mat U, \mat V)})^{-1/2} \nonumber \\
	&\qquad \times \exp\left[ - \frac 1 2 \vec x^\tp \mat \Sigma_{(\mat U, \mat V)}^{-1} \vec x \right] \nonumber \\
	&= \pi^{-N} \exp\left[ - \left(\mat S_{(\mat U, \mat V)}^{-1} \vec x\right)^\tp \left(\mat S_{(\mat U, \mat V)}^{-1} \vec x\right) \right]\,,
\end{align}
and since the state is pure, we can also write down its position-space wave function (up to an arbitrary overall phase)
\begin{align}
\label{eq:psifromuv}
	\psi_{(\mat U,\mat V)}(\vec q) &= \pi^{-N/4} (\det \mat U)^{1/4} \exp \left[- \frac 1 2 \vec q^\tp (\mat U - i\mat V) \vec q \right]\,.
\end{align}
Notice that $\vec q$, $\vec p$, and~$\vec x = \left( \begin{smallmatrix} \vec q \\ \vec p \end{smallmatrix} \right)$ are c-number column vectors that correspond to their respective operator-valued counterparts.  Any of these four equations can be used to define a Gaussian pure state from any pair of $N \times N$ symmetric matrices, $\mat U$ and~$\mat V$, with~$\mat U > 0$ ensuring physicality of the state.  Equation~\eqref{eq:Suv} defines the Heisenberg quadrature variables~$\opvec x = \mat S_{(\mat U, \mat V)} \opvec x_0$ associated with the state in question, and Eq.~\eqref{eq:covmatuv} gives the (symmetrized) covariance matrix, from which the Wigner function, Eq.~\eqref{eq:Wignerfromuv}, may be readily obtained.  If one wishes to work with wave functions, then Eq.~\eqref{eq:psifromuv} can be used.  Inversion of these relations to find~$\mat U$ and~$\mat V$ is straightforward.  The ground state corresponds to~$\mat U = \mat \id$ and~$\mat V = \mat 0$.

\subsection{Gaussian pure states as undirected graphs with complex-weighted edges}
\label{subsec:gauss:graph}

The complex combination~$\mat U - i\mat V$ that appears in Eq.~\eqref{eq:psifromuv} is suggestive of a way to unify the two symmetric matrices that define a Gaussian pure state.  Instead, we will multiply this by~$i$ and define
\begin{align}
\label{eq:Zdef}
	\mat Z := \mat V + i \mat U
\end{align}
for reasons that will become clear shortly.  This complex, symmetric matrix~$\mat Z$ is only useful, of course, if it has nice transformation properties under Gaussian unitary operations.  In fact, this is the case.  Defining
\begin{align}
\label{eq:SABCD}
	\mat S =
	\begin{pmatrix}
		\mat A	& \mat B	\\
		\mat C	& \mat D
	\end{pmatrix}
\end{align}
as the symplectic matrix corresponding to the evolution in question, if the initial state corresponds to~$\mat Z$ as above then the new state after the evolution will correspond to~\cite{Simon1988}
\begin{align}
\label{eq:Zprime}
	\mat Z' = (\mat C + \mat D \mat Z)(\mat A + \mat B \mat Z)^{-1}\,.
\end{align}
We will interpret~$\mat Z$ as the adjacency matrix for an undirected graph with complex-valued edge weights, thus providing our graph representation for any Gaussian pure state.  A rigorous derivation of this relation and of the unique map from Gaussian pure states to graphs is included in Appendix~\ref{app:derivations}.

\subsection{Gaussian graphs from expectation values of observables}
\label{subsec:gauss:expval}

An important operational question is, \emph{How can the graph for a Gaussian pure state be obtained from the statistics of measurements made on the system?}  To answer it, it is useful to consider the covariance matrix from Eq.~\eqref{eq:covmatuv}.  We can immediately extract~$\mat U$ from the upper-left block:
\begin{align}
\label{eq:ufromsigma}
	\mat U &= (2 \cov \opvec q)^{-1} = \frac 1 2 \avgg { \opvec q \opvec q^\tp }^{-1} \,.
\end{align}
Once we have~$\mat U$, extracting~$\mat V$ from the upper-right block is straightforward:
\begin{align}
\label{eq:vfromsigma}
	\mat V &= \mat U \avgg { \{ \opvec q, \opvec p^\tp \} } \nonumber \\
	&= \frac 1 2 \avgg { \opvec q \opvec q^\tp }^{-1} \avgg { \{ \opvec q, \opvec p^\tp \} }\,.
\end{align}
Putting these together gives
\begin{align}
\label{eq:Zfromsigma}
	\mat Z &= \mat V + i\mat U \nonumber \\
	&= \frac 1 2 \avgg { \opvec q \opvec q^\tp }^{-1} \left( \avgg { \{ \opvec q, \opvec p^\tp \} } + i \mat \id \right) \nonumber \\
	&= \frac 1 2 \avgg { \opvec q \opvec q^\tp }^{-1} \left( \avgg { \{ \opvec q, \opvec p^\tp \} } + \avgg { [ \opvec q, \opvec p^\tp ] } \right) \nonumber \\
	&= \avgg { \opvec q \opvec q^\tp }^{-1} \avgg { \opvec q \opvec p^\tp }\,.
\end{align}
Equation~\eqref{eq:Zfromsigma} shows how to extract~$\mat Z$ from the expectation values of $\op q_j \op q_k$ and $\op q_j \op p_k$, with the latter obtainable from the observables~$(\op q_j \op p_k + \op p_k \op q_j)$ using the form on the second line.

\subsection{Approximate CV cluster states}
\label{subsec:gauss:CVCS}

The graph representative of a Gaussian pure state defined above is, in fact, the natural way to extend the graph representation of ideal CV cluster state to their finitely squeezed Gaussian approximants.  The canonical method for creating a CV cluster state~\cite{Menicucci2006} is to squeeze all modes as much as possible in the momentum quadrature and then to apply~$\opCZ[\mat A]$ in accord with a (real-weighted) graph~$\mat A$.  An ideal cluster state~$\ket{\psi_{\mat A}}$ from Eq.~\eqref{eq:psiA} is obtained by taking the limit of infinite initial squeezing on all the modes.  Let's see what this looks like in our formalism.

The symplectic transformation corresponding to the canonically generated CV cluster state consists of two parts: the initial single-mode squeezing and the controlled-$Z$ operations.  If we take all modes to be momentum-squeezed such that their variance is reduced by a factor of~$e^{-2r}$, followed by $\opCZ[\mat A]$, this corresponds to a total symplectic transformation of
\begin{align}
	\begin{pmatrix}
		\mat \id	& \mat 0	\\
		\mat A	& \mat \id
	\end{pmatrix}
	\begin{pmatrix}
		e^r \mat \id	& \mat 0 \\
		\mat 0		& e^{-r} \mat \id
	\end{pmatrix}
	\,.
\end{align}
Comparing this with Eq.~\eqref{eq:Suv}, we can immediately read off that~$\mat U = e^{-2r} \mat \id$ and~$\mat V = \mat A$, and we find that
\begin{align}
\label{eq:canonicalZ}
	\mat Z_r := \mat A + ie^{-2r} \mat \id
\end{align}
corresponds to an $r$-indexed family of approximate CV cluster states with graph~$\mat A$ since
\begin{align}
	\lim_{r \to \infty} \mat Z_r = \mat A\,.
\end{align}
However, there are many other families of Gaussian pure states that fit the bill, including one that will be useful later (in Section
~\ref{sec:apps}):
\begin{align}
\label{eq:Zalphafamily}
	\mat Z_\alpha = i \sech 2\alpha\, \mat I + \tanh 2\alpha\, \mat A\,,
\end{align}
which satisfies~$\lim_{\alpha \to \infty} \mat Z_\alpha = \mat A$.  Figure~\ref{fig:complexgraphs} illustrates the correspondence.

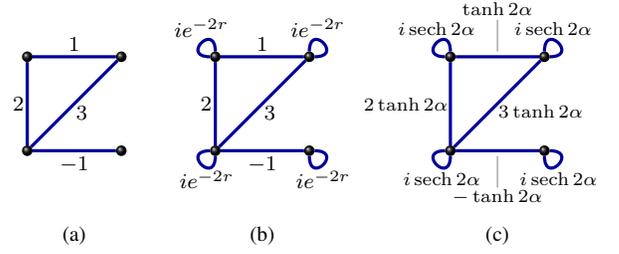
\begin{figure}[!tbp]
\begin{center}
\beginpgfgraphicnamed{CVCS_approx}
\begin{tikzpicture} [scale=1,label distance=-2pt, pin distance=2em]

	\footnotesize

	\def\orig{(0,0)}
	\def\canon{(2*\scale,0)}
	\def\hgraph{(4.5*\scale,0)}
	\def\labelloc{(-0.5*\scale,-1.9*\scale)}

	\def\scale{1.25}

	\path \orig +\labelloc node {(a)};
	\path \orig [outer sep=0pt]
		node [micro] (1) {}
		++(180:\scale) node [micro] (2) {}
		++(-90:\scale) node [micro] (3) {}
		++(0:\scale) node [micro] (4) {};

	\path (1)
		edge [unitlink] node [above=-1pt] {$1$} (2)
		edge [unitlink] node [below right=-2.5pt] {$3$} (3);

	\path (2)
		edge [unitlink] node [left=-2pt] {$2$} (3);
	
	\path (3)
		edge [unitlink] node [below=-1pt] {$-1$} (4);

	\path \canon +\labelloc node {(b)};
	\path \canon [outer sep=0pt]
		node [micro] (1) {}
		++(180:\scale) node [micro] (2) {}
		++(-90:\scale) node [micro] (3) {}
		++(0:\scale) node [micro] (4) {};

	\path (1)
		edge [unitlink,in=0,out=90, min distance=0.4cm] node [xshift=-2pt,above=-1pt] {$ie^{-2r}$} (1)
		edge [unitlink] node [above=-1pt] {$1$} (2)
		edge [unitlink] node [below right=-2.5pt] {$3$} (3);

	\path (2)
		edge [unitlink,in=90,out=180, min distance=0.4cm] node [above=-1pt] {$ie^{-2r}$} (2)
		edge [unitlink] node [left=-2pt] {$2$} (3);
	
	\path (3)
		edge [unitlink,in=180,out=270, min distance=0.4cm] node [xshift=2pt, below=-1pt] {$ie^{-2r}$} (3)
		edge [unitlink] node [below=-1pt] {$-1$} (4);
	
	\path (4)
		edge [unitlink,in=270,out=0, min distance=0.4cm] node [below=-1pt] {$ie^{-2r}$} (4);

	\path \hgraph +\labelloc node {(c)};
	\scriptsize
	\path \hgraph [outer sep=0pt]
		node [micro] (1) {}
		++(180:\scale) node [micro] (2) {}
		++(-90:\scale) node [micro] (3) {}
		++(0:\scale) node [micro] (4) {};

	\path (1)
		edge [unitlink,in=0,out=90, min distance=0.4cm] node [xshift=-2pt,above=0pt] {$i\sech 2\alpha$} (1)
		edge [unitlink] node [inner sep=0, outer sep=0.5pt, above, pin={[inner sep=1pt] 90:$\tanh 2\alpha $}] {} (2)
		edge [unitlink] node [below right=-2.5pt] {$3 \tanh 2\alpha $} (3);

	\path (2)
		edge [unitlink,in=90,out=180, min distance=0.4cm] node [above=0pt] {$i\sech 2\alpha$} (2)
		edge [unitlink] node [left=-2pt] {$2 \tanh 2\alpha$} (3);
	
	\path (3)
		edge [unitlink,in=180,out=270, min distance=0.4cm] node [xshift=2pt, below=0pt] {$i\sech 2\alpha$} (3)
		edge [unitlink] node [inner sep=0, outer sep=0.5pt, below, pin={[inner sep=0] -90:$-\tanh 2\alpha $}] {} (4);
	
	\path (4)
		edge [unitlink,in=270,out=0, min distance=0.4cm] node [below=0pt] {$i\sech 2\alpha$} (4);

\end{tikzpicture}
\endpgfgraphicnamed
\caption{Complex-weighted graph representation of approximate CV~cluster states.  Nodes no longer specify any particular input state on their own and are henceforth colored black to emphasize this distinction.  Instead, the state represented by the graph is entirely specified by the edge weights, which can now be complex.  The real part of the graph~$\mat V$ still has the same interpretation as in the original formalism---i.e.,~a collection of weights for the respective $\CZ$~gates applied between the linked nodes, while the imaginary part~$\mat U$ corresponds to initial multimode squeezing that only mixes~$\op q$s 
and~$\op p$s 
separately 
(see text).  (a)~The graph from Figure~\ref{fig:CVgraph} is reinterpreted in this formalism.  Because the graph has only real weights, it does not represent a physical state (since~$\mat U \not > 0$)---and rightly so since ideal CV cluster states are infinitely squeezed.  (b)~This is the graph for an approximate CV cluster state made by the canonical method~\cite{Menicucci2006,Gu2009}.  The imaginary-weighted self-loops indicate the amount of initial single-mode squeezing applied.  In the limit of large squeezing ($r \to \infty$), these physical Gaussian pure states limit to the ideal graph in~(a).  Notice that in this context, states with large squeezing in~$\op p$ are represented by black nodes with a vanishing imaginary self-loop weight, rather than by the white nodes of Figure~\ref{fig:CVgraph}.  (c)~Here is another approximate CV cluster state, distinct from that in~(b), that also limits to the ideal case~(a) when~$\alpha \to \infty$.  The states represented in~(b) and~(c) are physically distinct, but because they have the same large-squeezing behavior, they cannot be distinguished in the original graphical formalism (see~Figure~\ref{fig:CVgraph}).  This new formalism allows them to be distinguished at the graph level.}
\label{fig:complexgraphs}
\end{center}
\end{figure}

\subsection{\texorpdfstring{$\cH$}{H}-graph states}
\label{subsec:gauss:Hgraph}

$\cH$-graph states are generated by a multimode OPO pumped by a multifrequency pump beam and have a mathematical connection to CV cluster states, even though the nature of the two types of graph are different~\cite{Menicucci2007,Zaidi2008,Menicucci2008,Flammia2009}.  In our graphical formalism, these states correspond to graphs with purely imaginary weights.  Specifically,
\begin{align}
\label{eq:ZHgraph}
	\mat Z = ie^{-2 \alpha \mat G}\,,
\end{align}
where the real, symmetric matrix $\mat G$~is the $\cH$-graph for the state, and $\alpha$~is a unitless overall squeezing strength.  The term~``$\cH$-graph'' refers to the fact that $\mat G$~specifies the (linearized) Hamiltonian for the OPO that acts on the vacuum of the cavity modes in order to generate the state.  This Hamiltonian is defined by
\begin{align}
\label{eq:HfromG}
	\op{\cH}(\mat G) &= i \hbar \kappa \sum_{j,k} G_{jk} (\op a^\dag_j \op a^\dag_k - \op a_j \op a_k) \nonumber \\
	&= i\hbar \kappa (\opvec a^\herm \mat G \opvec a^\dag - \opvec a^\tp \mat G \opvec a) \nonumber \\
	&= \hbar \kappa (\opvec q^\tp \mat G \opvec p + \opvec p^\tp \mat G \opvec q) \,,
\end{align}
where $\kappa$~is a squeezing parameter per unit time~\cite{Menicucci2007}.  If this Hamiltonian is applied for time~$t$, acting on the vacuum state with~$\op U = \exp[-\tfrac i \hbar \op{\cH}(\mat G) t]$, then~$\alpha = 2 \kappa t$ in Eq.~\eqref{eq:ZHgraph}.

Most $\cH$-graph states (i.e.,~those with a full-rank $\mat G$) correspond to CV cluster states in the limit of large squeezing if one phase-shifts modes appropriately~\cite{Menicucci2007}.  Recent work shows that this method can be used, in a scalable fashion, to make many small, disconnected square-graph CV cluster states from a single OPO~\cite{Zaidi2008}.  More importantly, a single OPO can also be used to make very large QC-universal CV cluster states in a scalable fashion~\cite{Menicucci2008,Flammia2009}.

The connection between the $\cH$-graph~$\mat G$ and the Gaussian graph~$\mat Z$, as indicated by Eq.~\eqref{eq:ZHgraph}, is through the exponential map, which is generally a nontrivial operation on a graph.  However, when $\mat G$~is self-inverse ($\mat G^2 = \mat \id$), this connection simplifies greatly:
\begin{align}
\label{eq:ZHgraphselfinv}
	\mat G^2 = \mat \id \quad \Longrightarrow \quad \mat Z = i\cosh 2\alpha\, \mat \id - i\sinh 2\alpha\, \mat G
\end{align}
In this case, the Gaussian graph~$\mat Z$ is just a rescaled version of~$\mat G$ with additional self-loops.  We will explore the close connection between $\cH$-graphs of this form and CV~cluster states of the form of Eq.~\eqref{eq:Zalphafamily} in Section~\ref{sec:apps}.

\subsection{Complex nullifiers}
\label{subec:gauss:nullifier}

In Reference~\olcite{Gu2009}, the real-valued nullifier formalism is used to illustrate the effect of quadrature measurements on ideal CV cluster states.  (In an optical setting, this corresponds to homodyne detection.)  Rules for updating the nullifiers under such measurements are derived, but they only apply to ideal CV cluster states.  In addition, the resulting state after the measurement frequently is only an ideal CV cluster state up to phase shifts (if at all) and thus cannot be represented by any real-weighted CV cluster-state graph.  Zhang has implemented these rules as a graph transformation~\cite{Zhang2010}, but the same restriction to ideal CV cluster states still applies.  Here we extend the nullifier transformation rules from Reference~\olcite{Gu2009} to graph transformation rules for quadrature measurements made on any Gaussian pure state, including approximate CV cluster states, thus generalizing the results from Reference~\cite{Zhang2010}.  We do this by first generalizing the real-valued nullifier formalism to complex-valued nullifiers, which can be used to represent all Gaussian pure states.

The nullifier formalism for CV cluster states, given by Eq.~\eqref{eq:psiAnull}, can be extended to all Gaussian pure states using the simple replacement of the CV cluster-state graph~$\mat A$ with the Gaussian graph~$\mat Z$:
\begin{align}
\label{eq:psiZnull}
	(\opvec p - \mat Z \opvec q) \ket {\psi_{\mat Z}} &=
	\begin{pmatrix}
		-\mat Z	& \mat \id
	\end{pmatrix}
	\opvec x
	\op U_{\mat Z} \ket {0} \nonumber \\
	&=
	\op U_{\mat Z}
	\begin{pmatrix}
		-\mat Z	& \mat \id
	\end{pmatrix}
	\mat S_{\mat Z}
	\opvec x
	\ket {0} \nonumber \\
	&=
	\op U_{\mat Z} \mat U^{1/2} (\opvec p - i \opvec q) \ket 0 \nonumber  \\
	&= \vec 0\,,
\end{align}
where we have used Eq.~\eqref{eq:Suv} to plug in for~$\mat S_{\mat Z} =\mat S_{(\mat U, \mat V)}$, and we note that $\opvec p - i \opvec q = -i \sqrt 2 \opvec a$, a vector of operators annihilating the ground state.  Notice that the \emph{nullifier vector}~$\opvec p - \mat Z \opvec q$ is not unique for a given graph since any c-number matrix~$\mat M$, acting from the left, will generate a new vector of nullifiers that are also satisfied by the state:
\begin{align}
\label{eq:psiZnullother}
	(\mat M \opvec p - \mat M \mat Z \opvec q) \ket {\psi_{\mat Z}} &= 0\,.
\end{align}
The action of the matrix~$\mat M$ represents forming a new nullifier vector from linear combinations of the original nullifiers.  If, in addition, $\mat M$~is invertible, then this new nullifier vector also uniquely defines the state.

We can take the fact that $\opvec p - i \opvec q$ is proportional to a vector of annihilation operations and make the analogy more explicit.  Let's define
\begin{align}
\label{eq:aZ}
	\opvec a_{\mat Z} &:= \frac {i} {\sqrt 2} \mat U^{-1/2} (\opvec p - \mat Z \opvec q)\,, \\
\label{eq:aZdag}
	\opvec a_{\mat Z}^\dag &:= \frac {-i} {\sqrt 2} \mat U^{-1/2} (\opvec p - \mat Z^* \opvec q)\,,
\end{align}
which have the usual commutation relations
\begin{align}
\label{eq:aZcomm}
	[\opvec a_{\mat Z}, \opvec a_{\mat Z}^\herm] &= \frac 1 2 \mat U^{-1/2}
	\begin{pmatrix}
		-\mat Z	& \mat \id
	\end{pmatrix}
	[ \opvec x, \opvec x^\tp ]
	\begin{pmatrix}
		-\mat Z^*	\\
		\mat \id
	\end{pmatrix}
	\mat U^{-1/2} \nonumber \\
	&= \frac 1 2 \mat U^{-1/2}
	\begin{pmatrix}
		-\mat Z	& \mat \id
	\end{pmatrix}
	i\mat \Omega
	\begin{pmatrix}
		-\mat Z^*	\\
		\mat \id
	\end{pmatrix}
	\mat U^{-1/2} \nonumber \\
	&= \mat \id \,,
\end{align}
along with $[\opvec a_{\mat Z}, \opvec a_{\mat Z}^\tp] = [\opvec a_{\mat Z}^\dag, \opvec a_{\mat Z}^\herm] = \mat 0$.  In the case of the ground state, $\mat Z = i\mat \id$, and this expression reduces to the usual result ($\opvec a_{\mat Z} = \opvec a$).  In the more general case, these operators can be used to derive a Hamiltonian for which the associated graph~$\mat Z$ is the ground state (see Section~\ref{subsec:apps:ground}).

With this notation in hand, we can calculate
\begin{align}
\label{eq:covZ}
	\cov (\opvec p - \mat Z \opvec q) &= \cov (-i\sqrt 2 \mat U^{1/2} \opvec a_{\mat Z}) \nonumber \\
	&= \mat U^{1/2} \avg{ \{\opvec a_{\mat Z}^\dag, \opvec a_{\mat Z}^\tp \} } \mat U^{1/2} \nonumber \\
	&= \mat U^{1/2} \left\{ \avg{ 2 \opvec a_{\mat Z}^\dag \opvec a_{\mat Z}^\tp } + \avg{ [\opvec a_{\mat Z}, \opvec a_{\mat Z}^\herm ] } \right\} \mat U^{1/2} \nonumber \\
	&= \mat U\,,
\end{align}
where we used Eqs.~\eqref{eq:psiZnull} and~\eqref{eq:aZcomm} to obtain the last line.  For approximate CV cluster states, It's also useful to calculate the covariance matrix of just the real part of~$\mat Z$ (namely,~$\mat V$):
\begin{align}
\label{eq:covv}
	\cov [\opvec p - \mat V \opvec q] &=
	\begin{pmatrix}
		-\mat V	& \mat \id
	\end{pmatrix}
	(\cov \opvec x)
	\begin{pmatrix}
		-\mat V	\\
		\mat \id
	\end{pmatrix}
	\nonumber \\
	&= \frac 1 2
	\begin{pmatrix}
		\mat 0	& \mat \id
	\end{pmatrix}
	\begin{pmatrix}
		\mat U^{-1}	& \mat 0	\\
		\mat 0		& \mat U
	\end{pmatrix}
	\begin{pmatrix}
		\mat 0	\\
		\mat \id
	\end{pmatrix}
	\nonumber \\
	&= \frac 1 2 \mat U \,,
\end{align}
where we have used the explicit form for the covariance matrix---the second line of Eq.~\eqref{eq:covmatuv}.  Comparing this expression with Eq.~\eqref{eq:approxCVCS}, we get a nice interpretation of $\mat Z$'s real and imaginary parts: $\mat V$~is the graph of the ideal CV cluster state approximated by~$\mat Z$, and $\frac 1 2 \mat U$~is the error in the approximation---now expressed in quantitative terms as the covariance matrix of the nullifiers~$\opvec p - \mat V \opvec q$.  For this interpretation to make any sense, of course, $\mat U$ must be small.  Since $\mat U > 0$, we can use the trace to say that $\frac 1 2 \tr \mat U$~is the magnitude of the \emph{approximation error} in approximating the ideal CV cluster state~$\mat V$ using~$\mat Z$.  This trace corresponds to the sum of the variances of each of the nullifiers:
\begin{align}
	\frac 1 2 \tr \mat U = \sum_j \bra {\psi_{\mat Z}} \bigl( \op p_j - \sum_k V_{jk} \op q_k  \bigr)^2 \ket {\psi_{\mat Z}}\,,
\end{align}
which should be compared with Eq.~\eqref{eq:approxCVCSexplicit}.

This extension of the real nullifier formalism for real-weighted, ideal CV cluster states to a complex nullifier formalism for complex-weighted, physical CV cluster states---as well as all other Gaussian pure states---corresponds to a similar generalization on the level of the stabilizer operators, which, in the CV case, describe position and momentum shifts in phase space~\cite{Gu2009,vanLoock2007}. As we shall continue focusing on the nullifiers here, we give a brief derivation of the corresponding generalized stabilizer representation for Gaussian pure states in Appendix~\ref{app:stabilizers}.

\section{Transformation rules}
\label{sec:transrules}

In this section, we derive the rules for updating a graph after a local Gaussian unitary transformation is performed on one or more of the modes or when a quadrature measurement is made.  Since an arbitrary $n$-local transformation can be obtained by composing 1- and 2-local transformations, we will treat those cases with additional care.  We begin by deriving the affects of the symplectic transformations on the underlying matrix $\mat Z$, and following that, we make connection with the Gaussian unitary transformations they represent, and we illustrate these transformations graphically.

\subsection{\texorpdfstring{$n$}{n}-Local Gaussian unitary operations}

Consider a Gaussian pure state on $n+p$ modes.  It's graph~$\mat Z$ can be written in block form as
\begin{align}
\label{eq:nlocalblocks}
	\mat Z =
	\begin{pmatrix}
		\mat T	&	\mat R^\tp 	\\
		\mat R	&	\mat W
	\end{pmatrix}
	\,,
\end{align}
where $\mat T$~is $n \times n$, $\mat R$~is $p \times n$, and $\mat W$~is $p \times p$.  The symmetric matrix~$\mat T$ is the adjacency matrix for the induced subgraph of~$\mat Z$ formed by considering only the $n$~modes in question.  The other symmetric matrix~$\mat W$ represents the induced subgraph corresponding to the untouched modes.  $\mat R$, of course, represents the connection between the two sets of nodes.

Without loss of generality, we can represent an arbitrary $n$-local operation with the symplectic matrix
\begin{align}
	\mat S =
	\begin{pmatrix}
		\mat A	&	\mat 0	&	\mat B	&	\mat 0	\\
		\mat 0	&	\mat \id	&	\mat 0	&	\mat 0	\\
		\mat C	&	\mat 0	&	\mat D	&	\mat 0	\\
		\mat 0	&	\mat 0	&	\mat 0	&	\mat \id
	\end{pmatrix}
	\,,
\end{align}
and
\begin{align}
	\mat S_\text{$n$-local} =
	\begin{pmatrix}
		\mat A & \mat B \\
		\mat C & \mat D
	\end{pmatrix}
\end{align}
is the symplectic operation for just the $n$ target nodes.  Then, defining
\begin{align}
\label{eq:Jdef}
	\mat J = (\mat A + \mat B \mat T)^{-1}\,,
\end{align}
we follow Eq.~(\ref{eq:Zprime}) and write the transformed matrix as
\begin{align}
	\mat Z' &= \left[ \begin{pmatrix} \mat C & \mat 0 \\ \mat 0 & \mat 0 \end{pmatrix} + \begin{pmatrix} \mat D & \mat 0 \\ \mat 0 & \mat \id \end{pmatrix} \mat Z \right] \left[ \begin{pmatrix} \mat A & \mat 0 \\ \mat 0 & \mat \id \end{pmatrix} + \begin{pmatrix} \mat B & \mat 0 \\ \mat 0 & \mat 0 \end{pmatrix} \mat Z \right]^{-1} \nonumber \\
	&=
	\begin{pmatrix}
		\mat C + \mat D\mat T	&	\mat D\mat R^\tp 	\\
		\mat R		&	\mat W
	\end{pmatrix}
	\begin{pmatrix}
		\mat A + \mat B\mat T	&	\mat B\mat R^\tp 	\\
		\mat 0		&	\mat \id
	\end{pmatrix}
	^{-1} \nonumber \\
	&=
	\begin{pmatrix}
		\mat C + \mat D\mat T	&	\mat D\mat R^\tp 	\\
		\mat R		&	\mat W
	\end{pmatrix}
	\begin{pmatrix}
		\mat J	&	-\mat J\mat B\mat R^\tp 	\\
		\mat 0			&	\mat \id
	\end{pmatrix}
	\nonumber \\
	&=
	\begin{pmatrix}
		(\mat C + \mat D\mat T)\mat J	&	-(\mat C + \mat D\mat T)\mat J\mat B\mat R^\tp + \mat D\mat R^\tp 	\\
		\mat R\mat J			&	\mat W - \mat R\mat J\mat B\mat R^\tp
	\end{pmatrix}
	\,.
\end{align}
This matrix must be symmetric, so we know that the upper-right block must simply be $\mat J^\tp\mat R^\tp $, but let's work it out and see why.  The first fact we'll need is that $(\mat C + \mat D\mat T)(\mat A + \mat B\mat T)^{-1}$ is symmetric, since it represents a Gaussian operation on a valid graph (recall that $\mat T$ is symmetric and that $\Im \mat T > 0$).  Therefore,
\begin{align}
	&\mathrel{\phantom{=}} -(\mat C + \mat D\mat T) \mat J\mat B\mat R^\tp + \mat D\mat R^\tp \nonumber \\
	&= -\mat J ^\tp(\mat C + \mat D\mat T)^\tp \mat B\mat R^\tp + \mat D\mat R^\tp \nonumber \\
	&= \mat J ^\tp \bigl[ -(\mat C + \mat D\mat T)^\tp \mat B + \mat J ^{-\tp} \mat D \bigr] \mat R^\tp \nonumber \\
	&= \mat J^\tp \bigl[ -\mat C^\tp \mat B - \mat T\mat D^\tp \mat B + \mat A^\tp \mat D + \mat T\mat B^\tp \mat D \bigr] \mat R^\tp \,.
\end{align}
The properties of a symplectic matrix~\cite{Arvind1995} include $\mat A^\tp \mat D - \mat C^\tp \mat B = \mat \id$ and $\mat D^\tp \mat B = \mat B^\tp \mat D$.  Therefore, the quantity in brackets is equal to~$\mat \id$, and the $n$-local-transformed matrix is
\begin{align}
\label{eq:nlocaltransformedZ}
	&\mat Z' = \\
	&\begin{pmatrix}
		(\mat C + \mat D\mat T)(\mat A + \mat B\mat T)^{-1}	&	(\mat A + \mat B\mat T)^{-\tp}\mat R^\tp 	\\
		\mat R(\mat A + \mat B\mat T)^{-1}			&	\mat W - \mat R(\mat A + \mat B\mat T)^{-1}\mat B\mat R^\tp
	\end{pmatrix}
	. \nonumber
\end{align}
This has a nice interpretation.  The upper-left block is just the transformation that would result from applying the $n$-local Gaussian to the nodes in question, without any connection to other nodes.  The bottom-right block reflects the fact that any changes to the other nodes' induced subgraph are mediated solely by connections with the target nodes and in a \emph{additive} fashion on~$\mat W$.  These are the two key observations---the induced subgraph on the target nodes cares not about connections to other nodes, and the effect on adjacent nodes is only through connection with the target nodes.  The off-diagonal blocks illustrate the action of the transformation on the edges connecting the two sets.  After discussing the important special cases of 1- and 2-local transformations, we will illustrate these graphically to provide additional insight.

\subsection{Local Gaussian unitary operations}
\label{subsec:transrules:LG}

In the case of (1-)local Gaussian operations, so-called \emph{LG~operations}, with $2 \times 2$ symplectic matrix
\begin{align}
	\mat S_\text{LG} =
	\begin{pmatrix}
		a	&	b		\\
		c	&	d
	\end{pmatrix}
	\,,
\end{align}
this transformation is just
\begin{align}
\label{eq:ZunderLG}
	\mat Z =
	\begin{pmatrix}
		t		&	\vec r^\tp 	\\
		\vec r	&	\mat W
	\end{pmatrix}
	\xmapsto{\mat S_{\mathrm{LG}}}
	\begin{pmatrix}
		\displaystyle{\frac {c + d t} {a + b t}}		&	\displaystyle{\frac {\vec r^\tp } {a + b t}}		\\
		\displaystyle{\frac {\vec r} {a + b t}}		&	\mat W - \displaystyle{\frac {b \vec r \vec r^\tp } {a + b t}}
	\end{pmatrix}
	= \mat Z'\,.
\end{align}
From these results, we will calculate (below) local graph transformation rules corresponding to elementary LG operations.

\subsection{2-Local Gaussian operations}

Arbitrary Gaussian operations can be constructed out of 1-local (LG) and 2-local (2LG) operations alone.  In fact, given the availability of all LGs, only a single fiducial 2LG is needed to construct any Gaussian operation~\cite{Lloyd1999}.  The most theoretically simple 2LG is just the $\CZ$ gate.  This operation just adds a real constant---proportional to the strength of the interaction---to the edge in question.

Considering the applicability of our results to optical implementations, we shall also consider another operation (or actually, another class of operations), since the $\CZ$ gate is difficult to implement experimentally~\cite{Yurke1985}.  This class will consist of beamsplitter interactions.  Specifically, we consider only a fiducial type of photon-number-conserving interaction.  This 2LG interaction can be used to model any beamsplitter when combined with appropriate phase shifts on the input and output modes:
\begin{align}
\label{eq:SBS}
	\mat S_{\mathrm{BS}}(\theta) &=
	\begin{pmatrix}
		\mat R_\theta	&	\mat 0	\\
		\mat 0		&	\mat R_\theta
	\end{pmatrix}
	=
	\begin{pmatrix}
		\cos \theta	&	-\sin \theta	&	0		&	0		\\
		\sin \theta	&	\cos \theta	&	0		&	0		\\
		0		&	0		&	\cos \theta	&	-\sin \theta	\\
		0		&	0		&	\sin \theta	&	\cos \theta
	\end{pmatrix}
	\,,
\end{align}
where $\sin \theta$ is the amplitude reflectivity of the beamsplitter~\cite{Walls2008}.  This form is particularly simple because $\mat B=\mat C=\mat 0$ and $\mat A=\mat D=\mat R_\theta$, which gives $(\mat A+\mat B\mat T)^{-1} = \mat R_\theta^\tp = \mat R_{-\theta}$.  This being the only quantity that affects nodes outside the target set, this is particularly convenient. 
The transformation of~$\mat T$ is given by
\begin{widetext}
\begin{align}
\label{eq:TprimeBS}
	\mat T' = \mat R_\theta \mat T \mat R_\theta^\tp 
	&=
	\begin{pmatrix}
		\cos \theta	&	-\sin \theta		\\
		\sin \theta	&	\cos \theta	
	\end{pmatrix}	
	\begin{pmatrix}
		T_{11}	&	T_{12}	\\
		T_{12}	&	T_{22}
	\end{pmatrix}	
	\begin{pmatrix}
		\cos \theta	&	\sin \theta		\\
		-\sin \theta	&	\cos \theta	
	\end{pmatrix}
	\nonumber \\
	&= 
	\begin{pmatrix}
		T_{11} \cos^2 \theta + T_{22} \sin^2 \theta - T_{12} \sin 2\theta		& T_{12} \cos 2\theta + \tfrac 1 2 (T_{11} - T_{22}) \sin 2\theta	\\
		T_{12} \cos 2\theta + \tfrac 1 2 (T_{11} - T_{22}) \sin 2\theta		& T_{22} \cos^2 \theta + T_{11} \sin^2 \theta + T_{12} \sin 2\theta
	\end{pmatrix}
	\,.
\end{align}
\end{widetext}
The final result for the beamsplitter transformation is:
\begin{align}
	\mat Z' &=
	\begin{pmatrix}
		\mat T'	&	\mat R_\theta \mat R^\tp 	\\
		\mat R \mat R_\theta^\tp			&	\mat W
	\end{pmatrix}
	\\
	&=
	\begin{pmatrix}
		\mat R_\theta	&	\mat 0	\\
		\mat 0		&	\mat \id
	\end{pmatrix}
	\begin{pmatrix}
		\mat T	&	\mat R^\tp 	\\
		\mat R	&	\mat W
	\end{pmatrix}
	\begin{pmatrix}
		\mat R_\theta^\tp	&	\mat 0	\\
		\mat 0			&	\mat \id
	\end{pmatrix}
	\nonumber \,,
\end{align}
where $\mat T'$~is defined in Eq.~\eqref{eq:TprimeBS}, and we have included the second (expanded) form to show the formally simple effect that beamsplitter interactions of the form~$\mat S_{\mathrm{BS}}(\theta)$ have on Gaussian graphs.  (As previously mentioned, modeling a physical beamsplitter may additionally involve phase shifts, which are local Gaussian operations.)

\subsection{Quadrature measurements}
\label{subec:transrules:quadmeas}

According to the rules of Reference~\olcite{Gu2009}, the first thing to do when considering quadrature measurements is to define the nullifier that corresponds to the measurement outcome.  This is the new nullifier that the post-measurement state must satisfy due to projection onto the measurement basis.  (This rule applies even if the measurement is destructive.)  Next, we choose an appropriate invertible matrix~$\mat M$ such that the entries in the new (ideal) nullifier vector~$\mat M \opvec p - \mat M \mat V \opvec q$ are such that only one of them fails to commute with the nullifier corresponding to the measurement (something which can always be done).  This new nullifier vector also uniquely defines the pre-measurement state (since $\mat M$~is invertible), 
but because only one of its entries (i.e.,~a single nullifier) fails to commute with the measurement nullifier, all the remaining ones will also be nullifiers for the post-measurement state as well.  The one that fails to commute is therefore discarded and replaced by the measurement nullifier to form the post-measurement nullifier vector.

Everything said in the previous paragraph remains valid when considering the complex nullifiers from Section~\ref{subec:gauss:nullifier} instead of the ideal ones from Reference~\olcite{Gu2009}.  We will therefore focus on one specific example: $\op q$-measurements.  Such a measurement on node~$j$ with outcome~$s_j$ means that the new state has $\op q_j - s_j$~as one of its nullifiers (we assume the detector noise is negligible).  However, since we are neglecting displacements, the post-measurement state will instead have $\op q_j$~as the measurement nullifier.  The usual nullifier vector for a Gaussian graph~$\mat Z$ is~$\opvec p - \mat Z \opvec q$, which is already in a fortunate form since all of the nullifiers  commute with~$\op q_j$ except the~$j$th one.  This measurement on an ideal CV cluster state corresponds to deletion of the measured node from the cluster, along with all of its links.  The effect is the same on Gaussian graphs and can be seen in the nullifier formalism in that all references to~$\op p_j$ are gone, and linear combinations of the post-measurement nullifiers can be used to delete all references to~$\op q_j$, as well (since $\op q_j$~alone is one of the nullifiers in the set).  As an action on the adjacency matrix~$\mat Z$, this corresponds to deleting its $j$th row and column.

Another useful quadrature measurement is a measurement of~$\op p$.  For an ideal CV cluster state, such a measurement deletes the corresponding node but \emph{preserves the links} between neighboring nodes, up to local phase shifts.  Because of the phase shifts, the resulting state is often impossible to represent as an ideal CV cluster-state graph (although it still has a nullifier representation), but we can do it easily for approximate CV cluster states as a two-step process: (1)~perform an inverse Fourier transform on the node to be measured, and then (2)~perform a $\op q$~measurement as described above, thereby deleting it and its links from the graph.  This is not equivalent in total to a simple disconnection, however, because the phase shift generates additional connections in the neighborhood of the measured node before that node is disconnected.  It is through this mechanism that a measurement of~$\op p$ preserves links in the graph while deleting the measured node~\cite{Gu2009}.  Measurements of the general quadrature~$\op q \cos \theta + \op p \sin \theta$, which can be used to perform Gaussian dynamics\footnote{For universal QC, the ability to measure in a non-Gaussian basis is required~\cite{Gu2009}.  In an optical context, this can be achieved through photon counting.  Because the resulting state is no longer Gaussian, such measurements are not incorporated into this formalism.} on encoded CV quantum information in the cluster~\cite{Ukai2009,Gu2009}, are represented analogously, with a phase shift by~$\theta$ replacing the inverse Fourier transform before the $\op q$~measurement.

\subsection{Graph transformation rules}
\label{subsec:transrules:rules}

Here we illustrate the transformation rules described above as actions on the adjacency matrix~$\mat Z$ into rules transforming the associated graph for several examples.  In all cases, the original~$\mat Z$ is given in block form according to Eq.~\eqref{eq:nlocalblocks}, while Equations~\eqref{eq:ZunderLG} and~\eqref{eq:nlocaltransformedZ} provide the transformation laws for local and 2-local operations, respectively.  We will focus on representative local transformations, then the 2-local interactions discussed previously, and finally quadrature measurements.


\graphrule{Displacement}---The graph rule for displacements in phase space is trivial: do nothing to the graph.  The Gaussian graph only represents the noise properties of the state, which are unaffected by overall displacements.

\graphrule{Local shear}---The easiest nontrivial local operation to represent in this formalism is a local shear in phase space:~$\exp(\tfrac i 2 g \op q^2)$.  The corresponding local symplectic matrix is
\begin{align}
\label{eq:SLGshear}
	\mat S_\mathrm{LG} = \mat S_\mathrm{shear}(g) :=
	\begin{pmatrix}
		1	&	0		\\
		g	&	1
	\end{pmatrix}
	\,,
\end{align}
resulting in the simple transformation
\begin{align}
	\mat Z \xmapsto{\mat S_{\mathrm{shear}}(g)}
	\begin{pmatrix}
		t + g		&	\vec r^\tp	\\
		\vec r	&	\mat W
	\end{pmatrix}
	\,,
\end{align}
where $\mat Z$ is defined in Eq.~\eqref{eq:ZunderLG}.  Notice that the only edge affected by this transformation is the self-loop on the node experiencing the shear (and note that $\mat W$~is not shown):
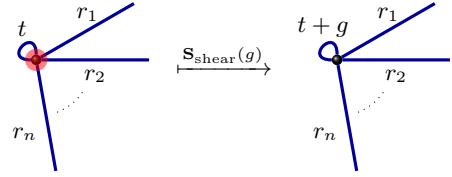
\begin{center}
\beginpgfgraphicnamed{shear}
\begin{tikzpicture} [scale=1]
	\small

	\def\orig{(0,0)}
	\def\final{(4,0)}
	\def\mapsymb{(2.5,0)}
	
	\def\scale{1.5}
	
	\path \orig
		+(30:\scale) coordinate (r1) 
		+(0:\scale) coordinate (r2)
		+(-80:\scale) coordinate (rn)
		[draw,dotted] +(-35:0.5*\scale) arc (-35:-70:0.5*\scale);
	\path \orig node [micro] (t) {}
		edge [unitlink] node [above] {$r_1$} (r1)
		edge [unitlink] node [below] {$r_2$} (r2)
		edge [unitlink] node [below left] {$r_n$} (rn)
		edge [unitlink, loop, in=90,out=180,min distance=0.4cm] node [above] {$t$} (t);

	\draw [nodehighlight] (t) circle (\micronodesize);

	\node at \mapsymb {$\xmapsto{\mat S_{\mathrm{shear}}(g)}$};

	\path \final
		+(30:\scale) coordinate (r1) 
		+(0:\scale) coordinate (r2)
		+(-80:\scale) coordinate (rn)
		[draw,dotted] +(-35:0.5*\scale) arc (-35:-70:0.5*\scale);
	\path \final node [micro] (t) {}
		edge [unitlink] node [above] {$r_1$} (r1)
		edge [unitlink] node [below] {$r_2$} (r2)
		edge [unitlink] node [below left] {$r_n$} (rn)
		edge [unitlink, loop, in=90,out=180,min distance=0.4cm] node [above] {$t+g$} (t);
\end{tikzpicture}
\endpgfgraphicnamed
\end{center}

\graphrule{Single-mode squeezing}---Squeezing (i.e.,~reducing the variance) in~$\op p$ with squeezing parameter~$r>0$ 
is equivalent to squeezing in~$\op q$ with squeezing parameter~$-r$ and is effected by the unitary operator~$\exp[-\tfrac i 2 r (\op q \op p + \op p \op q)] = \exp[\tfrac 1 2 r(a^{\dag2} - a^2)]$ and represented by the local symplectic matrix
\begin{align}
\label{eq:SLGsqueeze}
	\mat S_\mathrm{LG} = \mat S_\mathrm{squeeze}(r) :=
	\begin{pmatrix}
		e^r		&	0		\\
		0		&	e^{-r}
	\end{pmatrix}
\end{align}
and transforms graphs in a particularly simple way:
\begin{align}
	\mat Z \xmapsto{\mat S_{\mathrm{squeeze}}(r)}
	\begin{pmatrix}
		e^{-2r} t		&	e^{-r} \vec r^\tp	\\
		e^{-r} \vec r	&	\mat W
	\end{pmatrix}
	\,.
\end{align}
Notice that this transformation only affects the edges attached to the node being squeezed; there is no effect on the neighborhood of the affected node (represented by~$\mat W$, not shown):
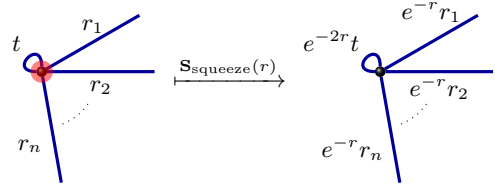
\begin{center}
\beginpgfgraphicnamed{squeeze}
\begin{tikzpicture} [scale=1]
	\small 

	\def\orig{(0,0)}
	\def\final{(4.5,0)}
	\def\mapsymb{(2.5,0)}

	\def\scale{1.5}

	\path \orig
		+(30:\scale) coordinate (r1) 
		+(0:\scale) coordinate (r2)
		+(-80:\scale) coordinate (rn)
		[draw,dotted] +(-35:0.5*\scale) arc (-35:-70:0.5*\scale);
	\path \orig node [micro] (t) {}
		edge [unitlink] node [above] {$r_1$} (r1)
		edge [unitlink] node [below] {$r_2$} (r2)
		edge [unitlink] node [below left] {$r_n$} (rn)
		edge [unitlink, loop, in=90,out=180,min distance=0.4cm] node [above left=-1pt] {$t$} (t);

	\draw [nodehighlight] (t) circle (\micronodesize);

	\node at \mapsymb {$\xmapsto{\mat S_{\mathrm{squeeze}}(r)}$};

	\path \final
		+(30:\scale) coordinate (r1) 
		+(0:\scale) coordinate (r2)
		+(-80:\scale) coordinate (rn)
		[draw,dotted] +(-35:0.5*\scale) arc (-35:-70:0.5*\scale);
	\path \final node [micro] (t) {}
		edge [unitlink] node [above=4pt] {$e^{-r} r_1$} (r1)
		edge [unitlink] node [below=-2pt] {$e^{-r} r_2$} (r2)
		edge [unitlink] node [below left=-1pt] {$e^{-r} r_n$} (rn)
		edge [unitlink, loop, in=90,out=180,min distance=0.4cm] node [above left=-1pt] {$e^{-2r}t$} (t);
\end{tikzpicture}
\endpgfgraphicnamed
\end{center}

\graphrule{Phase shift}---The transformations above only affect edges attached to the node being acted upon (the ``active" node).  In order to make a change to the $\mat W$-matrix, however, which represents the rest of the graph, we need something else.  A simple operation that accomplishes this goal and has the advantage of being very easy to implement in many experiments is the phase shift (a.k.a.\ phase delay), which simply rotates the phase plane of the active node by an angle~$\theta$.  The unitary operator is~$\exp[-\tfrac i 2 \theta (\op q^2 + \op p^2) ] = e^{-i\theta / 2} \exp(-i \theta \op a^\dag \op a)$, where the overall phase~$e^{-i\theta / 2}$ can be ignored.  The corresponding local symplectic matrix is just a rotation matrix:\footnote{There is some ambiguity in the sign used in the definition of the phase shift.  The convention we are using is consistent with References~\olcite{Menicucci2007,Menicucci2008,Flammia2009}.}
\begin{align}
\label{eq:SLGphaseshift}
	\mat S_\mathrm{LG} = \mat S_\mathrm{phase}(\theta) :=
	\begin{pmatrix}
		\cos \theta	&	\sin \theta		\\
		-\sin \theta	&	\cos \theta
	\end{pmatrix}
	\,.
\end{align}
The associated transformation is
\begin{align}
	\mat Z \xmapsto{\mat S_{\mathrm{phase}}(\theta)}
	\begin{pmatrix}
		\dfrac {- \sin \theta + t \cos \theta} {\cos \theta + t \sin \theta}	& \dfrac {\vec r^\tp }  {\cos \theta + t \sin \theta}	\\
		\dfrac {\vec r} {\cos \theta + t \sin \theta}	& \mat W - \dfrac { \sin \theta\, \vec r \vec r^\tp} {\cos \theta + t \sin \theta}
	\end{pmatrix}
\end{align}
Notice that phase shifts on one node can be used to induce additional links within its neighborhood (consider the case where~$W_{12} = 0$ in what follows; also note that $c_\theta = \cos \theta$ and~$s_\theta = \sin \theta$):%
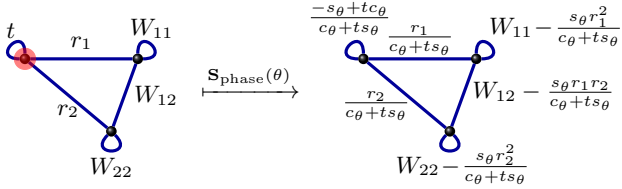
\begin{center}
\beginpgfgraphicnamed{phase}
\begin{tikzpicture} [scale=1]
	\small 

	\def\orig{(0,0)}
	\def\final{(4.5,0)}
	\def\mapsymb{(3,-0.3)}

	\def\scale{1.5}

	\path \orig
		+(0:\scale) node [micro] (r1) {}
		+(-40:\scale) node [micro] (r2) {};
	\path \orig node [micro] (t) {}
		edge [unitlink] node [above] {$r_1$} (r1)
		edge [unitlink] node [below] {$r_2$} (r2)
		edge [unitlink, loop, in=90,out=180,min distance=0.4cm] node [above=-1pt] {$t$} (t);

	\path (r1) 
		edge [unitlink, loop, in=0,out=90,min distance=0.4cm] node [above=-1pt] {$W_{11}$} (r1)
		edge [unitlink] node [right] {$W_{12}$} (r2);
		
	\path (r2) 
		edge [unitlink, loop, in=-45,out=-135,min distance=0.4cm] node [below=-1pt] {$W_{22}$} (r2);

	\draw [nodehighlight] (t) circle (\micronodesize);

	\node at \mapsymb {$\xmapsto{\mat S_{\mathrm{phase}}(\theta)}$};

	
	\path \final
		+(0:\scale) node [micro] (r1) {}
		+(-40:\scale) node [micro] (r2) {};

	\path \final node [micro] (t) {}
		edge [unitlink] node [above=-2pt] {$\frac {r_1}  {c_\theta + t s_\theta}$} (r1)
		edge [unitlink] node [below left=-5pt] {$\frac {r_2}  {c_\theta + t s_\theta}$} (r2)
		edge [unitlink, loop, in=90,out=180,min distance=0.4cm] node [above] {$\frac {- s_\theta + t c_\theta} {c_\theta + t s_\theta}$} (t);

	\path (r1) 
		edge [unitlink, loop, in=0,out=90,min distance=0.4cm] node [above right=-4pt] {$W_{11}\! - \! \frac {s_\theta r_1^2} {c_\theta + t s_\theta}$} (r1)
		edge [unitlink] node [right] {$W_{12} - \frac {s_\theta r_1 r_2} {c_\theta + t s_\theta}$} (r2);
		
	\path (r2) 
		edge [unitlink, loop, in=-45,out=-135,min distance=0.4cm] node [xshift=-2em, below right=-6pt] {$W_{22} \! -\! \frac {s_\theta r_2^2} {c_\theta + t s_\theta}$} (r2);
		
\end{tikzpicture}
\endpgfgraphicnamed
\end{center}
Any single-mode Gaussian unitary operation can be represented as a graph transformation by appropriately concatenating the rules for squeezing and phase shifting~\cite{Braunstein2005}.


At this point it is useful to mention that all of these rules agree with the graph rules derived by Zhang~\cite{Zhang2008a} in the limit of $\Im \mat Z \to \mat 0$, but only when the initial and final graphs remain finite in this limit.  For example, the Fourier transform,\footnote{Unlike the sign ambiguity for the definition of the phase shift, the Fourier transform is fixed by the requirement that measuring in~$\op p$ is the same as applying the inverse Fourier transform~$\op F^\dag$ and then measuring in~$\op q$~\cite{Menicucci2006}.} which corresponds to the phase shift
\begin{align}
	\mat F := \mat S_{\mathrm{phase}}(-\tfrac \pi 2) = \begin{pmatrix}
		0	& -1	\\
		1	& 0
	\end{pmatrix}\,,
\end{align}
gives
\begin{align}
\label{eq:Faction}
	\mat Z \xmapsto{\; \mat F \;}
	\begin{pmatrix}
		- t^{-1}		&	-t^{-1} \vec r^\tp		\\
		-t^{-1} \vec r	&	\mat W - t^{-1} \vec r \vec r^\tp
	\end{pmatrix}
\end{align}
and is represented as follows:
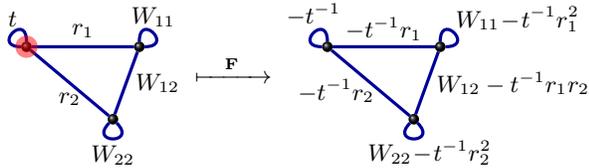
\begin{center}
\beginpgfgraphicnamed{Fourier}
\begin{tikzpicture} [scale=1]
	\small 

	\def\orig{(0,0)}
	\def\final{(4,0)}
	\def\mapsymb{(2.75,-0.3)}

	\def\scale{1.5}

	\path \orig
		+(0:\scale) node [micro] (r1) {}
		+(-40:\scale) node [micro] (r2) {};
	\path \orig node [micro] (t) {}
		edge [unitlink] node [above] {$r_1$} (r1)
		edge [unitlink] node [below] {$r_2$} (r2)
		edge [unitlink, loop, in=90,out=180,min distance=0.4cm] node [above=-1pt] {$t$} (t);

	\path (r1) 
		edge [unitlink, loop, in=0,out=90,min distance=0.4cm] node [above=-1pt] {$W_{11}$} (r1)
		edge [unitlink] node [right] {$W_{12}$} (r2);
		
	\path (r2) 
		edge [unitlink, loop, in=-45,out=-135,min distance=0.4cm] node [below=-1pt] {$W_{22}$} (r2);

	\draw [nodehighlight] (t) circle (\micronodesize);

	\node at \mapsymb {$\xmapsto{\quad \mat F \quad}$};

	
	\path \final
		+(0:\scale) node [micro] (r1) {}
		+(-40:\scale) node [micro] (r2) {};

	\path \final node [micro] (t) {}
		edge [unitlink] node [above=-1pt] {$-t^{-1} r_1$} (r1)
		edge [unitlink] node [below left=-5pt] {$-t^{-1} r_2$} (r2)
		edge [unitlink, loop, in=90,out=180,min distance=0.4cm] node [above=-1pt] {$-t^{-1}$} (t);

	\path (r1) 
		edge [unitlink, loop, in=0,out=90,min distance=0.4cm] node [above right=-3pt] {$W_{11}\! - \! t^{-1} r_1^2$} (r1)
		edge [unitlink] node [right] {$W_{12}-t^{-1} r_1 r_2$} (r2);
		
	\path (r2) 
		edge [unitlink, loop, in=-45,out=-135,min distance=0.4cm] node [xshift=-2em,below right=-3pt] {$W_{22} \! -\! t^{-1} r_2^2$} (r2);

\end{tikzpicture}
\endpgfgraphicnamed
\end{center}
This has no corresponding rule in the ideal limit if $t = 0$, i.e.,~if the active node has no self-loop, which is the case for most ideal CV cluster states of interest.  Nevertheless, applying the Fourier transform twice has a corresponding rule in the ideal limit, even when $t=0$, since $\mat F^2 = -\mat \id$:
\begin{align}
	\mat Z \xmapsto{\; \mat F^2 \;}
	\begin{pmatrix}
		t		&	- \vec r^\tp		\\
		- \vec r	&	\mat W
	\end{pmatrix}
	\,.
\end{align}
This rule corresponds exactly to Figure~5 in Ref.~\olcite{Zhang2008a}.  Similarly, the rule for local squeezing corresponds to Figure~6 in that reference (and is valid for $t=0$), while Figures~3 and~4 in that work can also be derived from the rules given here.

\graphrule{Controlled-$Z$ gate}---Similar to the local shear operation discussed above, the controlled-$Z$~gate~$\opCZ(g) = \exp(ig\op q_1 \op q_2)$ is the easiest 2LG operation to represent in the graphical formalism.  The corresponding symplectic matrix is
\begin{align}
	\mat S_\mathrm{2LG} = \mat S_{\CZ}(g) :=
	\begin{pmatrix}
   	 	1	& 0	& 0	& 0	\\
		0	& 1	& 0	& 0	\\
		0	& g	& 1	& 0	\\
		g	& 0	& 0	& 1
	\end{pmatrix}
	\,,
\end{align}
resulting in the simple transformation
\begin{align}
	\mat Z \xmapsto{\mat S_{\CZ}(g)}
	\begin{pmatrix}
    		T_{11}	& T_{12}+g	& \multicolumn{2}{c}{\multirow{2}{*}{$\mat R^\tp$}}  	\\
		T_{21}+g	& T_{22}		& \quad 	& \quad \\
		\multicolumn{2}{c}{\multirow{2}{*}{$\mat R$}}			& \multicolumn{2}{c}{\multirow{2}{*}{$\mat W$}}	\\
		&
	\end{pmatrix}
    	\,,
\end{align}
The graphical representation of this is simply to add~$g$ to the edge weight between the two active nodes:
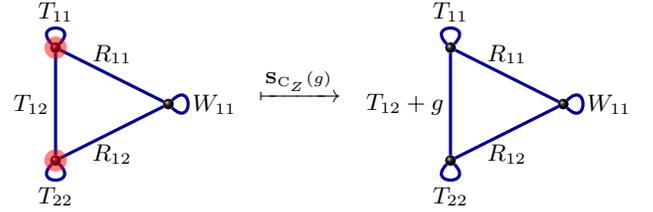
\begin{center}
\beginpgfgraphicnamed{CZ}
\begin{tikzpicture} [scale=1]
	\small 

	\def\orig{(0,0)}
	\def\final{(5.25,0)}
	\def\mapsymb{(3.25,-0.5)}

	\def\scale{1.5}
	\def\smallfactor{0.75}

	\path \orig node [micro] (t1) {}
		+(\scale,-0.5*\scale) node [micro] (r1a) {}
		++(0,-1*\scale) node [micro] (t2) {}
		;
		
	\path (t1)
		edge [unitlink] node [above=1pt] {$R_{11}$} (r1a)
		edge [unitlink, loop, in=45,out=135,min distance=0.4cm] node [above=-1pt] {$T_{11}$} (t1);
	\path (t2)
		edge [unitlink] node [below=1pt] {$R_{12}$} (r1a)
		edge [unitlink, loop, in=-45,out=-135,min distance=0.4cm] node [below=-1pt] {$T_{22}$} (t2);
	\path (t1) edge [unitlink] node [left=-1pt] {$T_{12}$} (t2);

	\path (r1a)
		edge [unitlink, loop, in=-45,out=45,min distance=0.4cm] node [right=-2pt] {$W_{11}$} (r1a); 

	\draw [nodehighlight] (t1) circle (\micronodesize);
	\draw [nodehighlight] (t2) circle (\micronodesize);

	\node at \mapsymb {$\xmapsto{\mat S_{\CZ}(g)}$};

	\path \final node [micro] (t1) {}
		+(\scale,-0.5*\scale) node [micro] (r1a) {}
		++(0,-1*\scale) node [micro] (t2) {}
		;
		
	\path (t1)
		edge [unitlink] node [above=1pt] {$R_{11}$} (r1a)
		edge [unitlink, loop, in=45,out=135,min distance=0.4cm] node [above=-1pt] {$T_{11}$} (t1);
	\path (t2)
		edge [unitlink] node [below=1pt] {$R_{12}$} (r1a)
		edge [unitlink, loop, in=-45,out=-135,min distance=0.4cm] node [below=-1pt] {$T_{22}$} (t2);
	\path (t1) edge [unitlink] node [left=-1pt] {$T_{12}+g$} (t2);
	\path (r1a)
		edge [unitlink, loop, in=-45,out=45,min distance=0.4cm] node [right=-2pt] {$W_{11}$} (r1a);

\end{tikzpicture}
\endpgfgraphicnamed
\end{center}
All such gates commute, and thus they can be performed in any order or even simultaneously.  Despite this theoretical simplicity, their experimental difficulty~\cite{Yurke1985} suggests developing another rule for a canonical 2LG operation.

\graphrule{Beamsplitter}---The beamsplitter interaction~$\mat S_{\text{BS}}(\theta)$ defined in Eq.~\eqref{eq:SBS}, which corresponds to the unitary~$\exp[-i \theta (\op q_1 \op p_2 - \op q_2 \op p_1)] = \exp[-\theta (\op a_1 \op a_2^\dag - \op a_2 \op a_1^\dag) ]$ and whose action on~$\mat Z$ is given by Eq.~\eqref{eq:TprimeBS}, has the following graph transformation rule (also note that~$\mat W$ is unaffected and that~$c_\theta = \cos \theta$ and~$s_\theta = \sin \theta$):
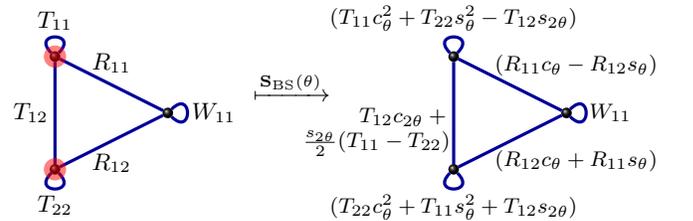
\begin{center}
\beginpgfgraphicnamed{BS}
\begin{tikzpicture} [scale=1]
	\small 

	\def\orig{(0,0)}
	\def\final{(5.3,0)}
	\def\mapsymb{(3.15,-0.4)}

	\def\scale{1.5}
	\def\smallfactor{0.75}

	\path \orig node [micro] (t1) {}
		+(\scale,-0.5*\scale) node [micro] (r1a) {}
		++(0,-1*\scale) node [micro] (t2) {}
		;
		
	\path (t1)
		edge [unitlink] node [above=1pt] {$R_{11}$} (r1a)
		edge [unitlink, loop, in=45,out=135,min distance=0.4cm] node [above=-1pt] {$T_{11}$} (t1);
	\path (t2)
		edge [unitlink] node [below=1pt] {$R_{12}$} (r1a)
		edge [unitlink, loop, in=-45,out=-135,min distance=0.4cm] node [below=-1pt] {$T_{22}$} (t2);
	\path (t1) edge [unitlink] node [left=-1pt] {$T_{12}$} (t2);

	\path (r1a)
		edge [unitlink, loop, in=-45,out=45,min distance=0.4cm] node [right=-2pt] {$W_{11}$} (r1a); 

	\draw [nodehighlight] (t1) circle (\micronodesize);
	\draw [nodehighlight] (t2) circle (\micronodesize);

	\node at \mapsymb {$\xmapsto{\mat S_{\mathrm{BS}}(\theta)}$};

	\path \final node [micro] (t1) {}
		+(\scale,-0.5*\scale) node [micro] (r1a) {}
		++(0,-1*\scale) node [micro] (t2) {}
		;
		
	\footnotesize
	
	\path (t1)
		edge [unitlink] node [above=0pt,anchor=south west,xshift=-9pt] {$(R_{11} c_\theta - R_{12} s_\theta)$} (r1a)
		edge [unitlink, loop, in=45,out=135,min distance=0.4cm] node [above=-1pt] {$(T_{11} c_\theta^2 + T_{22} s_\theta^2 - T_{12} s_{2\theta})$} (t1);
	\path (t2)
		edge [unitlink] node [below=0pt,anchor=north west,xshift=-9pt] {$(R_{12} c_\theta + R_{11} s_\theta)$} (r1a)
		edge [unitlink, loop, in=-45,out=-135,min distance=0.4cm] node [below=-1pt] {$(T_{22} c_\theta^2 + T_{11} s_\theta^2 + T_{12} s_{2\theta})$} (t2);
	\path (t1) edge [unitlink] node [text width=2cm, left=-3pt, yshift=-0.9em] {$\qquad \;\; T_{12} c_{2\theta} + {\tfrac {s_{2\theta}} {2} (T_{11} - T_{22})}$} (t2);
	\path (r1a)
		edge [unitlink, loop, in=-45,out=45,min distance=0.4cm] node [right=-2pt] {$W_{11}$} (r1a);

\end{tikzpicture}
\endpgfgraphicnamed
\end{center}
A useful special case of the above rule is the 50:50 beamsplitter, for which $\theta = \frac \pi 4$:
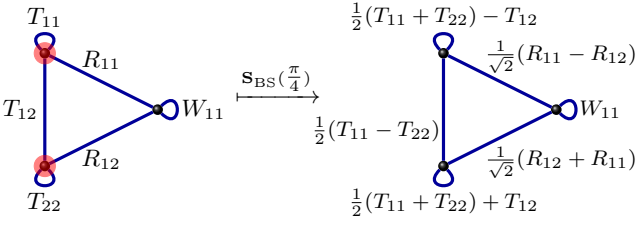
\begin{center}
\beginpgfgraphicnamed{BS50}
\begin{tikzpicture} [scale=1]
	\small 

	\def\orig{(0,0)}
	\def\final{(5.3,0)}
	\def\mapsymb{(3.1,-0.4)}

	\def\scale{1.5}
	\def\smallfactor{0.75}

	\path \orig node [micro] (t1) {}
		+(\scale,-0.5*\scale) node [micro] (r1a) {}
		++(0,-1*\scale) node [micro] (t2) {}
		;
		
	\path (t1)
		edge [unitlink] node [above=1pt] {$R_{11}$} (r1a)
		edge [unitlink, loop, in=45,out=135,min distance=0.4cm] node [above=-1pt] {$T_{11}$} (t1);
	\path (t2)
		edge [unitlink] node [below=1pt] {$R_{12}$} (r1a)
		edge [unitlink, loop, in=-45,out=-135,min distance=0.4cm] node [below=-1pt] {$T_{22}$} (t2);
	\path (t1) edge [unitlink] node [left=-1pt] {$T_{12}$} (t2);

	\path (r1a)
		edge [unitlink, loop, in=-45,out=45,min distance=0.4cm] node [right=-2pt] {$W_{11}$} (r1a); 

	\draw [nodehighlight] (t1) circle (\micronodesize);
	\draw [nodehighlight] (t2) circle (\micronodesize);

	\node at \mapsymb {$\xmapsto{\mat S_{\mathrm{BS}}(\tfrac \pi 4)}$};

	\path \final node [micro] (t1) {}
		+(\scale,-0.5*\scale) node [micro] (r1a) {}
		++(0,-1*\scale) node [micro] (t2) {}
		;
		
	\footnotesize
	
	\path (t1)
		edge [unitlink] node [above=0pt,anchor=south west,xshift=-9pt] {$\tfrac {1} {\sqrt 2} (R_{11} - R_{12})$} (r1a)
		edge [unitlink, loop, in=45,out=135,min distance=0.4cm] node [above=-2pt] {$\tfrac 1 2 (T_{11} + T_{22}) - T_{12}$} (t1);
		
	\path (t2)
		edge [unitlink] node [below=0pt,anchor=north west,xshift=-9pt] {$\tfrac {1} {\sqrt 2} (R_{12} + R_{11})$} (r1a)
		edge [unitlink, loop, in=-45,out=-135,min distance=0.4cm] node [below=-2pt] {$\tfrac 1 2 (T_{11} + T_{22}) + T_{12}$} (t2);
		
	\path (t1) edge [unitlink] node [left=-2pt, yshift=-1.0em] {$\tfrac 1 2 (T_{11} - T_{22})$} (t2);
	
	\path (r1a)
		edge [unitlink, loop, in=-45,out=45,min distance=0.4cm] node [right=-2pt] {$W_{11}$} (r1a);

\end{tikzpicture}
\endpgfgraphicnamed
\end{center}
Any multimode Gaussian operation can be represented as a graph transformation by appropriately concatenating the rules for arbitrary beamsplitters, single-mode squeezing, and phase shifts through the Bloch-Messiah decomposition~\cite{Braunstein2005}.

\graphrule{Measurement of~$\op q$, $\op p$, and other quadratures}---The rule for projective measurement of~$\op q$ is to delete the node and its links from the graph:
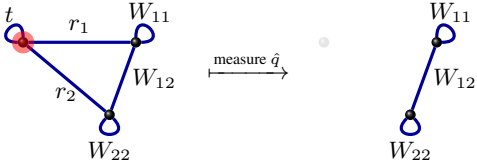
\begin{center}
\beginpgfgraphicnamed{qmeas}
\begin{tikzpicture} [scale=1]
	\small 

	\def\orig{(0,0)}
	\def\final{(4,0)}
	\def\mapsymb{(3,-0.3)}

	\def\scale{1.5}

	\path \orig
		+(0:\scale) node [micro] (r1) {}
		+(-40:\scale) node [micro] (r2) {};
	\path \orig node [micro] (t) {}
		edge [unitlink] node [above] {$r_1$} (r1)
		edge [unitlink] node [below] {$r_2$} (r2)
		edge [unitlink, loop, in=90,out=180,min distance=0.4cm] node [above=-1pt] {$t$} (t);

	\path (r1) 
		edge [unitlink, loop, in=0,out=90,min distance=0.4cm] node [above=-1pt] {$W_{11}$} (r1)
		edge [unitlink] node [right] {$W_{12}$} (r2);

	\path (r2) 
		edge [unitlink, loop, in=-45,out=-135,min distance=0.4cm] node [below=-1pt] {$W_{22}$} (r2);

	\draw [nodehighlight] (t) circle (\micronodesize);

	\node at \mapsymb {$\xmapsto{\text{measure~$\op q$}}$};

	
	\path \final
		+(0:\scale) node [micro] (r1) {}
		+(-40:\scale) node [micro] (r2) {};
	\path \final node [micro,very nearly transparent] (t) {}
	;

	\path (r1) 
		edge [unitlink, loop, in=0,out=90,min distance=0.4cm] node [above=-1pt] {$W_{11}$} (r1)
		edge [unitlink] node [right] {$W_{12}$} (r2);

	\path (r2) 
		edge [unitlink, loop, in=-45,out=-135,min distance=0.4cm] node [below=-1pt] {$W_{22}$} (r2);
		
\end{tikzpicture}
\endpgfgraphicnamed
\end{center}
This is the only measurement rule we need because, as shown in Section~\ref{subec:transrules:quadmeas}, the rule for measuring quadratures other than~$\op q$ is to phase shift the node to be measured so that a subsequent measurement of~$\op q$ is equivalent to the intended quadrature measurement on the original state.  For example, an inverse Fourier transform~$\mat F^{-1} = \mat S_{\mathrm{phase}}(\tfrac \pi 2)$ followed by a $\op q$~measurement implements a $\op p$~measurement on the original state.  The graph rule for this is just a concatenation of the rule for the phase shift followed by node deletion.

Because phase shifts can induce new links in the neighborhood of the shifted node (i.e.,~it can change~$\mat W$), measurements other than of~$\op q$ will in general preserve links between nodes that were previously mediated by the deleted node.  We show a special case of this by measuring~$\op p$ on the previous graph with $W_{12} = 0$:
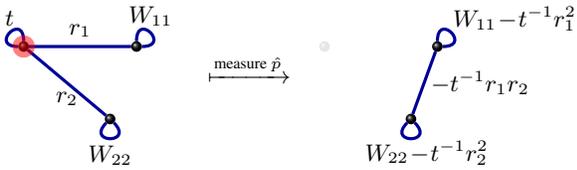
\begin{center}
\beginpgfgraphicnamed{pmeas}
\begin{tikzpicture} [scale=1]
	\small 

	\def\orig{(0,0)}
	\def\final{(4,0)}
	\def\mapsymb{(3,-0.3)}

	\def\scale{1.5}

	\path \orig
		+(0:\scale) node [micro] (r1) {}
		+(-40:\scale) node [micro] (r2) {};
	\path \orig node [micro] (t) {}
		edge [unitlink] node [above] {$r_1$} (r1)
		edge [unitlink] node [below] {$r_2$} (r2)
		edge [unitlink, loop, in=90,out=180,min distance=0.4cm] node [above=-1pt] {$t$} (t);

	\path (r1) 
		edge [unitlink, loop, in=0,out=90,min distance=0.4cm] node [above=-1pt] {$W_{11}$} (r1);
		
	\path (r2) 
		edge [unitlink, loop, in=-45,out=-135,min distance=0.4cm] node [below=-1pt] {$W_{22}$} (r2);

	\draw [nodehighlight] (t) circle (\micronodesize);

	\node at \mapsymb {$\xmapsto{\text{measure~$\op p$}}$};

	
	\path \final
		+(0:\scale) node [micro] (r1) {}
		+(-40:\scale) node [micro] (r2) {};
	\path \final node [micro,very nearly transparent] (t) {}
	;

	\path (r1) 
		edge [unitlink, loop, in=0,out=90,min distance=0.4cm] node [above right=-3pt] {$W_{11}\! - \! t^{-1} r_1^2$} (r1)
		edge [unitlink] node [right=-1pt] {$-t^{-1} r_1 r_2$} (r2);
		
	\path (r2) 
		edge [unitlink, loop, in=-45,out=-135,min distance=0.4cm] node [xshift=-2em,below right=-3pt] {$W_{22} \! -\! t^{-1} r_2^2$} (r2);

\end{tikzpicture}
\endpgfgraphicnamed
\end{center}
Notice that since the measured node mediated a connection between the two other nodes, measuring it in~$\op p$ preserved this connection in the form of a new edge connecting those two nodes directly.  Also notice that the strength of this new connection is proportional to both of the weights~$r_{1,2}$ in the original mediated connection.  These rules agree with Zhang's rules for quadrature measurements on ideal CV cluster states~\cite{Zhang2010}.

\section{Applications}
\label{sec:apps}

\subsection{Closest CV cluster state to a given Gaussian pure state}
\label{subsec:apps:closest}

Given the many different ways to make approximate CV cluster states~\cite{Zhang2006,Menicucci2006,vanLoock2007,Menicucci2007,Zaidi2008,Menicucci2008,Flammia2009,Menicucci2010}, a useful question to ask is, \emph{What is the closest CV cluster state approximated by a given Gaussian pure state?}  Imagine that you are given a system with a Gaussian graph~$\mat Z$ and you want to know if you can use it as a CV cluster state for one-way QC.  From Section~\ref{subec:gauss:nullifier}, we know that~$\mat V$ is the graph for the ideal CV cluster state approximated by~$\mat Z$ and that $\frac 1 2 \tr \mat U$ is the approximation error (which vanishes in the ideal case).  On the other hand, we also know that $\cH$-graph states have~$\mat V = \mat 0$ and a diverging~$\frac 1 2 \tr \mat U$ in the limit of large squeezing, yet they are useful for CV one-way~QC~\cite{Menicucci2007,Zaidi2008,Menicucci2008,Flammia2009}.  This na\"ive prescription is therefore not enough.

In fact, transforming these $\cH$-graph states into a useful form requires phase shifting nodes appropriately~\cite{Menicucci2007}, which, of course, also transforms the graph~$\mat Z \mapsto \mat Z'$.  The $\cH$-graph method, first proposed in Reference~\cite{Menicucci2007}, creates approximations to ideal CV cluster states with a bipartite graph.  The nodes of a bipartite graph can be colored such no edge links two nodes of the same color.  The prescription from Reference~\cite{Menicucci2007} for using these states requires first performing a Fourier transform (i.e.,~phase shift by~$-\tfrac \pi 2$) on some of the nodes and then using the resulting state as an ordinary CV cluster state (see Reference~\cite{Menicucci2007} for details).  The transformed~$\mat Z'$ has a nonzero~$\mat V'$ and a small approximation error (in the limit of large squeezing), which is the basis for the main result in that paper.  But is this prescription the best we can do?  Or might there be some other ideal CV cluster state obtained by a different phase shifting procedure that is better approximated by the original $\cH$-graph state?

We will address this question in the next subsection, but to answer it, we need to generalize our notion of what it means for a Gaussian graph to serve as a CV cluster-state graph.  Specifically, we must allow the freedom to phase shift any mode arbitrarily since this operation can be absorbed into the measurement-based protocol to be implemented on the state and is thus simply a redefinition of the quadratures and does not need to be actively implemented.  Our task then is to minimize the approximation error~$\frac 1 2 \tr \mat U'$ over all possible phase shifts (where the prime indicates the new graph after these phase shifts have been applied).  Once such a minimum is obtained, the resulting $\mat V'$~is the ``closest'' ideal CV cluster state approximated by~$\mat Z'$ and thus also the closest one approximated by the original~$\mat Z$.

The details of the calculation are somewhat involved, so we relegate them to Appendix~\ref{app:derivationclosest}, choosing to list here just the results.  Let
\begin{align}
\label{eq:Smattheta}
	\mat S_{\mat \theta} =
	\begin{pmatrix}
		\cos \mat \theta	&	\sin \mat \theta	\\
		-\sin \mat \theta	&	\cos \mat \theta		
	\end{pmatrix}
	\,,
\end{align}
where $\mat \theta = \diag (\theta_1, \dotsc, \theta_N)$.  This is the symplectic matrix representing the phase shifts~$\theta_j$ to be performed on each node~$j$.  In terms of the original graph~$\mat Z$, the new graph is
\begin{align}
\label{eq:Zprimeafterphases}
	\mat Z' = (-\sin \mat \theta + \cos \mat \theta\, \mat Z) (\cos \mat \theta + \sin \mat \theta\, \mat Z)^{-1}\,.
\end{align}
A sufficient condition for~$\frac 1 2 \tr \mat U'$ to be an \emph{extremum} (see Appendix~\ref{app:derivationclosest}) is that
\begin{align}
\label{eq:extremumcond}
	(\Im \mat Z'^2)_{jj} = ( \mat U' \mat V' )_{jj} = 0 \quad \forall j\,,
\end{align}
in other words, when corresponding rows of~$\mat U'$ and~$\mat V'$ are orthogonal.  A sufficient condition for such an extremum to also be a \emph{local minimum} (see Appendix~\ref{app:derivationclosest}) is that
\begin{align}
\label{eq:Hesspos}
	\Im [(\mat \id + \mat Z'^2) \circ \mat Z'] > 0\,,
\end{align}
where $\circ$~represents the Hadamard (entrywise) product of two matrices.  An interesting generalization of this condition is when~$\Im [(\mat \id + \mat Z'^2) \circ \mat Z'] \ge 0$, which is the best one can do when there is a continuous manifold of phases that all (locally) minimize~$\frac 1 2 \tr \mat U'$.  We will see an example of this in Section~\ref{subsec:apps:TMS} when we analyze the two-mode squeezed state.

We would like to stress that these are \emph{sufficient} conditions for \emph{local minima} only; some minima may not be able to be found this way, and not all minima found in this way may be global minima.  But in certain useful cases (such as the ones that follow), we can apply these results to provide evidence that we have found the closest CV cluster state (up to phase shifts) for a given graph.  With rigorous proofs of optimality left to future work, our purpose here is to illustrate a useful application of this graphical formalism.

\subsection{Analysis of the \texorpdfstring{$\cH$}{H}-graph method of construction}
\label{subsec:apps:Hopt}

As shown in Section~\ref{subsec:gauss:Hgraph}, all Gaussian pure states created using the $\cH$-graph method~\cite{Menicucci2007} necessarily satisfy~$\mat Z = i\mat U = ie^{-2 \alpha \mat G}$, where $\mat G$~is the $\cH$-graph that defines the multimode squeezing Hamiltonian, Eq.~\eqref{eq:HfromG}, that acts on the vacuum, and $\alpha > 0$~is an overall squeezing parameter.  To transform this state into an approximate CV cluster state, the $\cH$-graph method prescribes phase shifting particular nodes in accordance with the desired CV cluster-state graph~\cite{Menicucci2007}.  We will thus partition~$\mat Z$ into blocks as in Eq.~\eqref{eq:nlocalblocks} in accord with the partitioning of nodes to be shifted and nodes to be left alone:
\begin{align}
	\mat Z &=
	\begin{pmatrix}
		i \mat U_1		& i \mat U_2^\tp	\\
		i \mat U_2	& i \mat U_3
	\end{pmatrix}
\end{align}
We now perform a Fourier transform~$\mat F$ on the nodes corresponding to the upper-left block, using Eq.~\eqref{eq:Faction} for each or applying Eq.~\eqref{eq:nlocaltransformedZ} directly, resulting in 
\begin{align}
	\mat Z' &=
	\begin{pmatrix}
		i \mat U_1^{-1}			&	-\mat U_1^{-1} \mat U_2^\tp	\\
		-\mat U_2 \mat U_1^{-1}	&	i \mat U_3 - i \mat U_2 \mat U_1^{-1} \mat U_2^\tp
	\end{pmatrix}
	\,.
\end{align}
We want to know whether this procedure results in a minimum in the approximation error~$\frac 1 2 \tr \mat U'$.  This is difficult in general, but we will do so for a particularly useful case shortly.  For the moment, we can prove something weaker but still interesting: this particular choice of phase shifts results in an \emph{extremum} of~$\frac 1 2 \tr \mat U'$.  This is easily seen since every entry in~$\mat Z'$ is either purely real or purely imaginary.  Therefore $(\mat U' \mat V')_{jj} = 0$ for all~$j$, and thus, by Eq.~\eqref{eq:extremumcond}, $\frac 1 2 \tr \mat U'$~is an extremum.

In fact, this holds for \emph{any} set of phase shifts by multiples of~$\frac \pi 2$ on purely imaginary graphs, since the above construction didn't depend on any particular node(s) being selected for phase shifting, and shifting by~$\pi$ just applies negative signs, which doesn't change the real\slash imaginary nature of the entries.  Many of these will result in rather large values of the approximation error, but multiples of~$\tfrac \pi 2$ mean that these cases are still local extrema.  The prescription in Reference~\cite{Menicucci2007} explicitly chooses the nodes to be phase shifted so that the desired ideal CV cluster-state nullifiers vanish in the limit~$\alpha \to \infty$.  What was believed, but not shown explicitly, is that for any \emph{finite} value of~$\alpha$, the prescription \emph{minimized} these variances.  While we have still not shown this rigorously in the general case, we have provided additional evidence for this claim by showing that it results in an extremum of the sum of these variances.  We will be able to say more about particular examples, presented next.

\subsection{Two-mode squeezed state}
\label{subsec:apps:TMS}

The simplest nontrivial example of $\cH$-graph construction of a CV cluster state is embodied by the two-mode squeezed state, which results from downconversion in a nondegenerate OPO~\cite{Reid1988,Reid1989,Drummond1990,Ou1992a,Braunstein2005a}.  This procedure applies the Hamiltonian from Eq.~\eqref{eq:HfromG} with an $\cH$-graph~\cite{Menicucci2007}
\begin{align}
\label{eq:GTMS}
	\mat G =
	\begin{pmatrix}
		0	&	1	\\
		1	&	0
	\end{pmatrix}
	\,.
\end{align}
Using Eq.~\eqref{eq:ZHgraph}, the state created has%
\footnote{\label{foot:TMS}Alternatively, the same state can be generated by interfering a $\op q$-squeezed state~$Z=ie^{2\alpha}$ with a $\op p$-squeezed state~$Z=ie^{-2\alpha}$ at a 50:50 beamsplitter~\cite{Furusawa1998,vanLoock1999,Braunstein2005a}---i.e.,~by applying Eq.~\eqref{eq:SBS} with ${\theta = \tfrac \pi 4}$ to ${\mat Z = \left( \begin{smallmatrix} ie^{-2\alpha} & 0 \\ 0 & ie^{2\alpha} \end{smallmatrix} \right)}$.}
\begin{align}\label{eq:TMSS_adjMatrix}
	\mat Z_1 &= ie^{-2 \alpha \mat G} =
	\begin{pmatrix}
		i \cosh 2\alpha	&	-i \sinh 2\alpha	\\
		-i \sinh 2\alpha	&	i \cosh 2\alpha
	\end{pmatrix}
	\,,
\end{align}
where $\alpha > 0$ is an overall squeezing parameter, and the subscript is used because this state will be compared below to a canonically-generated two-mode CV cluster state.  This state is symmetric under exchange of the nodes, so we can choose to Fourier transform either one, either of which, using Eq.~\eqref{eq:Faction}, results in
\begin{align}
	\mat Z_1' &=
	\begin{pmatrix}
		i \sech 2\alpha	&	\tanh 2\alpha	\\
		\tanh 2\alpha	&	i \sech 2\alpha
	\end{pmatrix}
	\,.
\end{align}
This is a two-mode CV cluster state with weight~$\tanh 2\alpha$, which goes to~1 in the limit~$\alpha \to \infty$, and approximation error~$\frac 1 2 \tr \mat U_1' = \sech 2\alpha$, which vanishes in the same limit.  This means that~$\mat Z_1' \to \mat G$ in this limit, making trivial the connection between the generated CV cluster state and its generating~$\cH$-graph for the two-mode squeezed state.  This property is not generic---most $\cH$-graphs generate entirely different CV cluster states~\cite{Menicucci2007}---but a particular class of $\cH$-graphs admit this trivial connection.  This is discussed further in Section~\ref{subsec:apps:selfinv}.

For finite~$\alpha$, any combination of possible phase shifts by multiples of~$\frac \pi 2$ (see Section~\ref{subsec:apps:Hopt}) results in an approximation error of either~$\cosh 2\alpha$ or~$\sech 2\alpha$.  $\mat Z_1'$~corresponds to the latter---a (local) minimum in the approximation error.  Let's try to use Eq.~\eqref{eq:Hesspos} to verify this, though:
\begin{align}
	\Im [(\mat \id + \mat Z_1'^2) \circ \mat Z_1'] &= 2 \sech 2\alpha \tanh^2 2\alpha
	\begin{pmatrix}
		1	& 1	\\
		1	& 1
	\end{pmatrix}
	\,,
\end{align}
which has eignevalues~0 and~$4 \sech 2\alpha \tanh^2 2\alpha$, and thus only the weaker condition, $\Im [(\mat \id + \mat Z_1'^2) \circ \mat Z_1'] \ge 0$, is satisfied.  The zero eigenvalue in this case corresponds to the fact that any additional phase shift by~$\theta$ on the first node can be \emph{completely canceled} by an additional phase shift by~$-\theta$ on the second node.  This means that there is a one-parameter manifold of phase shifts that all result in~$\mat Z_1 \mapsto \mat Z_1'$, which has approximation error~$\sech 2\alpha$ and which is the minimum value obtainable.  Specifically, all phase shifts~$(\theta_1, \theta_2)$ satisfying~$\theta_1 + \theta_2 = -\frac \pi 2$ will result in minimal error as an approximation to an ideal CV cluster state.  In addition, a second manifold defined by~$\theta_1 + \theta_2 = +\frac \pi 2$ gives a similar minimum-error approximate CV cluster state but one which has weight~$-\tanh 2\alpha$.

We can also create a version of this state as a canonical CV cluster state.  In this case, using Eq.~\eqref{eq:canonicalZ} gives
\begin{align}
\label{eq:canontwomode}
	\mat Z_2 &=
	\begin{pmatrix}
		i e^{-2r}	&	1	\\
		1		&	i e^{-2r}
	\end{pmatrix}
	\,.
\end{align}
Both~$\mat Z_1'$ and~$\mat Z_2$ are complete graphs on 2 nodes with self-loops, and while they are identical in the infinite squeezing limit ($\alpha \to \infty, r \to \infty$), the weightings on each are different for any \emph{finite} amount of squeezing.  While real-weighted ideal CV cluster-state graphs fail to illustrate this, the complex graphical formalism captures the difference.

The importance of this difference comes from the relative efficiency of each method in its use of squeezing resources.  It is known that the canonically-generated CV cluster states (e.g.,~$\mat Z_2$) are inefficient in this sense because the resulting state has single-mode-squeezed marginals~\cite{vanLoock2007,Gu2009}.  An ``efficient'' state has~$\avg{\op q_j \op q_j} = \avg{\op p_j \op p_j}$ and~$\tfrac 1 2 \avg{ \{ \op q_j, \op p_j \} } = 0$ for all nodes~$j$, which means that all of the correlations are between quadratures variables of different systems~\cite{Adesso2006}.  Recalling Eq.~\eqref{eq:Wignerfromuv} for the covariance matrix in terms of~$\mat Z$, these requirements become
\begin{align}
\label{eq:efficient}
	(\mat U + \mat V \mat U^{-1} \mat V - \mat U^{-1})_{jj} = (\mat V \mat U^{-1})_{jj} = 0 \quad \forall j\,.
\end{align}
This fails for~$\mat Z_2$ because
\begin{align}
	(\mat U_2 + \mat V_2 \mat U_2^{-1} \mat V_2 - \mat U_2^{-1})_{jj} &= e^{-2r}\,,
\end{align}
which vanishes only in the infinite-squeezing limit~($r \to \infty$).  In contrast, Eq.~\eqref{eq:efficient} \emph{does} hold for~$\mat Z_1'$.  Thus, the $\cH$-graph method of constructing a two-mode CV cluster state is efficient in its use of squeezing resources, while the canonical method is not.  This fact cannot be seen in the real-weighted graphical formalism, but the complex formalism reveals it.

\subsection{Bipartite, self-inverse \texorpdfstring{$\cH$}{H}-graphs}
\label{subsec:apps:selfinv}

Notice that for the two-mode squeezed state created using the $\cH$-graph method, $\mat Z_1' \to \mat G$~in the limit~$\alpha \to \infty$ of large squeezing.  This is a general feature of $\cH$-graphs that are \emph{bipartite} and \emph{self-inverse}~\cite{Zaidi2008}.  Such graphs include some with square-lattice topology that are useful for universal one-way QC~\cite{Flammia2009}.  Recalling Eq.~\eqref{eq:ZHgraphselfinv}, when $\mat G$~is self-inverse (i.e.,~$\mat G^2 = \mat \id$), the resulting $\cH$-graph Gaussian pure state is
\begin{align}
\label{eq:ZHgraphselfinvrepeat}
	\mat Z = ie^{-2\alpha \mat G} = i \cosh 2\alpha\, \mat \id - i\sinh 2\alpha\, \mat G\,,
\end{align}
In this case, the Gaussian graph~$\mat Z$ is just a rescaled version of~$\mat G$ with additional self-loops.  When $\mat G$~is also bipartite, it can be written as
\begin{align}
	\mat G =
	\begin{pmatrix}
		\mat 0	& \mat G_0^\tp	\\
		\mat G_0	& \mat 0
	\end{pmatrix}
	\,,
\end{align}
where $\mat G_0$~is square and satisfies~$\mat G_0^\tp \mat G_0 = \mat G_0 \mat G_0^\tp = \mat \id$, giving
\begin{align}
	\mat Z =
	\begin{pmatrix}
		i\cosh 2\alpha\, \mat \id	& -i\sinh 2\alpha\, \mat G_0^\tp	\\
		-i\sinh 2\alpha\, \mat G_0	& i\cosh 2\alpha\, \mat \id
	\end{pmatrix}
	\,.
\end{align}
Performing a Fourier transform on the either set of nodes in the bipartition~\cite{Menicucci2007}, using Eq.~\eqref{eq:Faction}, gives
\begin{align}
	\mat Z' &=
	\begin{pmatrix}
		i\sech 2\alpha\, \mat \id	& \tanh 2\alpha\, \mat G_0^\tp	\\
		\tanh 2\alpha\, \mat G_0	& i\sech 2\alpha\, \mat \id
	\end{pmatrix}
	\nonumber \\
	&= i\sech 2\alpha\, \mat \id + \tanh 2\alpha\, \mat G
	\,,
\end{align}
which satisfies~$\mat Z' \to \mat G$ for large squeezing ($\alpha \to \infty$), corresponding to an ideal CV cluster state with the same graph as the $\cH$-graph~$\mat G$~\cite{Zaidi2008,Flammia2009}.  The two-mode squeezed state from Section~\ref{subsec:apps:TMS} is the simplest special case of this construction.

We already know from Section~\ref{subsec:apps:Hopt} that phase shifting by multiples of~$\frac \pi 2$ the nodes of a Gaussian pure state created using the $\cH$-graph method of construction results in approximate CV cluster states with extremal approximation error~$\frac 1 2 \tr \mat U$.  In this case, the possible values of this error are
\begin{align}
	\frac 1 2 \tr \mat U &= n \sech 2\alpha + (N_b-n) \cosh 2\alpha\,,
\end{align}
where $N_b$~is the total number of nodes in one of the bipartitions (equal to half the total number of nodes in the graph).  Clearly the minimum of these choices occurs when~$n=N_b$, corresponding (nonuniquely) to the original prescription~\cite{Menicucci2007}: Fourier transform all nodes in one of the bipartitions, and do nothing to those in the other set.  This results in~$\frac 1 2 \tr \mat U = N_b \sech 2\alpha$, which vanishes in the limit~$\alpha \to \infty$.

Let's try to verify Eq.~\eqref{eq:Hesspos} for this choice:
\begin{align}
\label{eq:Hessselfinv}
	\Im &[(\mat \id + \mat Z'^2) \circ \mat Z'] 
	\nonumber \\
	&= \Im \bigl[ (2\tanh^2 2\alpha\, \mat \id + 2i \sech 2\alpha \tanh 2\alpha\, \mat G) \nonumber \\
	&\qquad 
	\circ ( i\sech 2\alpha\, \mat \id + \tanh 2\alpha\, \mat G) \bigr]
	\nonumber \\
	&= 2 \sech 2\alpha \tanh^2 2\alpha
	\begin{pmatrix}
		\mat \id				& \mat G_0^\tp \circ \mat G_0^\tp	\\
		\mat G_0 \circ \mat G_0	& \mat \id
	\end{pmatrix}
	\nonumber \\
	&= 2 \sech 2\alpha \tanh^2 2\alpha\, (\mat \id + \mat G) \circ (\mat \id + \mat G) \ge 0
	\,,
\end{align}
where we have used the Schur product theorem ($\mat A \circ \mat B \ge 0$ if~$\mat A, \mat B \ge 0$), and~$\mat \id + \mat G \ge 0$ since $\mat G$~is self-inverse.  Therefore, $\Im [(\mat \id + \mat Z'^2) \circ \mat Z'] \ge 0$, which we also found independently for the two-mode squeezed state in Section~\ref{subsec:apps:TMS}.  In some cases of interest, like for the two-mode squeezed state, there will be a manifold of phase shifts that result in equivalent minimal-error CV cluster states made using a bipartite, self-inverse~$\cH$-graph, thus accounting for the vanishing eigenvalue(s) in Eq.~\eqref{eq:Hessselfinv}.  While this does not rigorously prove that the extremal value of the approximation error in the $\cH$-graph construction method for self-inverse, bipartite $\cH$-graphs~\cite{Zaidi2008,Flammia2009} is a global minimum (or even, strictly speaking, a local minimum), these calculations suffice to illustrate the usefulness of the complex graphical formalism and suggest further avenues of research in this area.

\subsection{GHZ state}
\label{subsec:apps:GHZ}

Also discussed in the literature~\cite{Menicucci2007,Pfister2007,Bradley2005,Pfister2004} is the CV Greenberger-Horne-Zeilinger~(GHZ) state, which can be made for any number of systems~$N$ using a complete $\cH$-graph\footnote{Alternatively, just like for the two-mode squeezed state (see footnote~\ref{foot:TMS}), the CV GHZ state can be made using single-mode squeezing and interferometry~\cite{vanLoock2000}.} with no self-loops:~$\mat G = \mat \id - \mat J$, where $\mat J$~is the $N \times N$~matrix of all ones.  We note that the diagonal of~$\mat G$ is not fixed for this state, but there are restrictions on it.  In order to obtain a fully squeezed state in the limit~$\alpha \to \infty$, an~$\cH$-graph must be full rank~\cite{Menicucci2007} (i.e.,~all eigenvalues must be nonzero).  In order to obtain GHZ entanglement, self-loops on the complete graph must be adjusted so that $\mat G$~has at least one eigenvalue of each sign~\cite{Pfister2007,Bradley2005,Pfister2004}.  Since the spectrum of~$\mat J$ is~$(N, 0, \dotsc, 0)$, any $\cH$-graph will suffice that is of the form~$\beta \mat \id - \mat J$ with $\beta > 0$.  
We choose a hollow~$\mat G$ (i.e.,~zero diagonal, corresponding to~$\beta=1$) for calculational convenience and because it has been studied specifically in the literature~\cite{Bradley2005,Pfister2004}.

Since~$\mat J^k = N^{k-1} \mat J$ for~$k \ge 1$, the Gaussian graph associated with~$\mat G$ is
\begin{align}
	\mat Z = i e^{-2\alpha} e^{2\alpha \mat J}
	&= i e^{-2\alpha} \left( \mat \id + \frac {\mat J} {N} \sum_{k=1}^\infty \frac {(2\alpha N)^k} {k!} \right) \nonumber \\
	&= i e^{-2\alpha} \left( \mat \id + \frac {e^{2\alpha N} - 1} {N} \mat J \right) \,,
\end{align}
As prescribed in Reference~\cite{Menicucci2007}, we wish to perform a Fourier transform on the first node using Eq.~\eqref{eq:Faction}.  To do this, we partition the graph as follows:
\begin{align}
	\mat Z &=
	\begin{pmatrix}
		t		& r		& r		& \cdots	& r		\\
		r		& W		& w		& \cdots	& w		\\
		r		& w		& W		& \ddots	& \vdots	\\
		\vdots	& \vdots	& \ddots	& \ddots	& w		\\
		r		& w		& \cdots	& w		& W
	\end{pmatrix}
	\,,
\end{align}
where
\begin{align}
	t = W &= \frac {i e^{-2\alpha}} {N} ( N + e^{2\alpha N} - 1 ) \nonumber \\
	&= \frac {i} {N \epsilon^{N-1}} [1 + (N-1) \epsilon^N] \,, \\
	r = w &= \frac {i e^{-2\alpha}} {N} ( e^{2\alpha N} - 1 ) \nonumber \\
	&= \frac {i} {N \epsilon^{N-1}} [1 - \epsilon^N] \,,
\end{align}
and where we have defined the small parameter~$\epsilon := e^{-2\alpha}$ because we are eventually interested in the limit~$\alpha \to \infty$.  After applying Eq.~\eqref{eq:Faction}, the resulting graph is
\begin{align}
	\mat Z' &=
	\begin{pmatrix}
		t'		& r'		& r'		& \cdots	& r'		\\
		r'		& W'		& w'		& \cdots	& w'		\\
		r'		& w'		& W'		& \ddots	& \vdots	\\
		\vdots	& \vdots	& \ddots	& \ddots	& w'		\\
		r'		& w'		& \cdots	& w'		& W'
	\end{pmatrix}
	\,,
\end{align}
for which we will evaluate each term exactly and also to the lowest nontrivial order in~$\epsilon$:
\begin{align}
	t' = -\frac 1 t &= \frac {iN\epsilon^{N-1}} {1 + (N-1)\epsilon^N} 
		\simeq iN\epsilon^{N-1}\,, \\
	r' = - \frac r t &= \frac {-1 + \epsilon^N} {1 + (N-1) \epsilon^N} 
		\simeq -1 + N\epsilon^N\,, \\
	W' = W - \frac {r^2} {t} &= i\epsilon \frac {2 + (N-2) \epsilon^N} {1 + (N-1)\epsilon^N} \simeq i 2\epsilon\,, \\
	w' = w - \frac {r^2} {t} &= i \epsilon \frac {1 - \epsilon^N} {1 + (N - 1) \epsilon^N} \simeq i \epsilon\,.
\end{align}
In the infinite-squeezing limit ($\alpha \to \infty$, which corresponds to~$\epsilon \to 0$), all terms vanish except~$r'$, resulting in an ideal CV cluster state with a star graph and~$-1$ for all edge weights:
\begin{align}
	\mat Z' &\to
	\begin{pmatrix}
		0		& -1		& \cdots	& -1		\\
		-1		& 0		& \cdots	& 0		\\
		\vdots	& \vdots	& \ddots	& \vdots	\\
		-1		& 0		& \cdots	& 0
	\end{pmatrix}
	\,.
\end{align}
The phase-shifted node is the center of the star.  A similar connection between complete graphs and star graphs is known for qubit cluster states~\cite{VandenNest2004,Hein2004}.  We hope this formalism will be a useful tool in generalizing results like these (for qubits) to the realm of CVs.

That a Gaussian pure state generated from a complete $\cH$-graph corresponds---after a Fourier transform on one of the nodes---to a star-graph CV cluster state in the infinite-squeezing limit is already known~\cite{Menicucci2007}, but the existing graph transformation rules for ideal CV cluster states~\cite{Zhang2008a} do not allow representation of the necessary Fourier transform operation.  What is new here---and what this construction illustrates---is a unified presentation that includes approximate CV cluster states and $\cH$-graphs (through the exponential map), plus intermediate Gaussian pure states and the rules for transforming between them, which captures \emph{all} of the details associated with finite squeezing wholly from within the graphical formalism.

\subsection{CV cluster states as ground states of a 2-body Hamiltonian}
\label{subsec:apps:ground}

Here we derive a (nonunique) Hamiltonian whose ground state is a particular CV cluster state\footnote{After the initial appearance of this work, but before its publication, another result appeared addressing this idea~\cite{Aolita2010}.  The analysis is limited to canonical CV cluster states~\cite{Menicucci2006} and uses the original nullifier formalism~\cite{Gu2009}, but the possible use of ground states of quadratic Hamiltonians for CV one-way QC is explored in much greater depth.} or, more generally, any given Gaussian graph~$\mat Z$.  The straightforward Hamiltonian to consider is
\begin{align}
	H[\mat Z] &= 2 \opvec a_{\mat Z}^\herm \opvec a_{\mat Z} 
\,,
\end{align}
which satisfies~$H[\mat Z] \ket {\psi_{\mat Z}} = 0$ by Eq.~\eqref{eq:psiZnull}.  This two-body Hamiltonian is also positive definite by construction, which means that $\ket {\psi_{\mat Z}}$~is its ground state.  Instead, however, we will use
\begin{align}
	H[\mat Z] &= (\opvec p - \mat Z \opvec q)^\herm (\opvec p - \mat Z \opvec q) \nonumber \\
	&= (\opvec p^\tp - \opvec q^\tp \mat Z^*) (\opvec p - \mat Z \opvec q) \nonumber \\
	&= \opvec p^\tp \opvec p - \opvec q^\tp \mat Z^* \opvec p - \opvec p^\tp \mat Z \opvec q + \opvec q^\tp \mat Z^* \mat Z \opvec q \nonumber \\
	&= \opvec p^\tp \opvec p - \opvec q^\tp (-i \mat U + \mat V)  \opvec p - \opvec p^\tp (i \mat U + \mat V) \opvec q \nonumber \\
	&\qquad+ \opvec q^\tp (\mat U^2 + \mat V^2 - i\mat U \mat V + i \mat V \mat U) \opvec q\,.
\end{align}
This two-body Hamiltonian, too, satisfies~$H[\mat Z] \ket {\psi_{\mat Z}} = 0$ by Eq.~\eqref{eq:psiZnull}, and it is positive definite by construction, which means that $\ket {\psi_{\mat Z}}$~is its ground state.  Reference~\olcite{Menicucci2007} proves that any CV cluster state with a bipartite graph is equivalent to an~$\cH$-graph state up to phase shifts of~$-\frac \pi 2$.  Restricting to these graphs simplifies the construction even further:
\begin{align}
	H[i\mat U] &= \opvec p^\tp \opvec p + i \opvec q^\tp \mat U  \opvec p - i \opvec p^\tp \mat U \opvec q + \opvec q^\tp \mat U^2 \opvec q\,.
\end{align}
Considering that
\begin{align}
	i \opvec q^\tp \mat U  \opvec p - i \opvec p^\tp \mat U \opvec q &= i \sum_{jk} (\op q_j U_{jk} \op p_k - \op p_j U_{jk} \op q_k) \nonumber \\
	&= i \sum_{jk} (\op q_j U_{jk} \op p_k - \op p_k U_{jk} \op q_j) \nonumber \\
	&= i \sum_{jk} [\op q_j, \op p_k] U_{jk} \nonumber \\
	&= -\sum_{j} U_{jj} \nonumber \\
	&= -\tr \mat U
\end{align}
is just an overall shift in energy (which can be ignored), we have the result that
\begin{align}
	H &= \opvec p^\tp \opvec p + \opvec q^\tp \mat U^2 \opvec q
\end{align}
has an $\cH$-graph-representable ground state (with $\mat G \propto -\log \mat U$).  Thus, for any desired CV cluster state with a bipartite graph~\cite{Menicucci2007}, some Hamiltonian of this form exists that has the desired state as its ground state (up to phase shifts).

\subsection{Bipartite entanglement}
\label{subsec:apps:entanglement}

A general $N$-mode Gaussian pure state may exhibit some form of multipartite entanglement.
In this case, it is useful to consider the entanglement for any bipartite splitting of the state. 
Given an $N$-mode Gaussian pure state, we split the modes into two sets, one with $n$ and the other one with $m$ modes,
$N=n+m$. We wish to calculate the entanglement between the two sets.  

Since the overall state is pure, we may use the entanglement entropy for this, which is simply the entropy of one of the subsystems with the other traced out. For a Gaussian pure state, the entanglement entropy is a simple function of the symplectic eigenvalues~\cite{Vidal2002} of the covariance matrix corresponding to the nodes in question.  The symplectic eigenvalues are similar to ordinary eigenvalues of a matrix, except that a symplectic product is taken between the matrix and its symplectic eigenvectors instead of the ordinary matrix product and the magnitude of the imaginary part is then taken.  Thus, while the matrix equation~$\mat M \vec v_j = \lambda_j \vec v_j$ defines the ordinary eigenvalues of the matrix~$\mat M$, the following equation defines the eigenvalues associated with the symplectic product:
\begin{align}
	\mat \Sigma \mat \Omega \vec v_j = \lambda_j \vec v_j\,.
\end{align}
Notice the presence of the symplectic form~$\mat \Omega$ in this relation.  Further note that the matrix $\mat \Sigma \mat \Omega$ is not Hermitian, since we have $(\mat \Sigma \mat \Omega)^\dagger = \mat \Omega^\tp \mat \Sigma = - \mat \Omega \mat \Sigma$, but both~$\mat \Sigma \mat \Omega$ and~$\mat \Omega \mat \Sigma$ have the same eigenvalues (since they are related by a similarity transformation).  Therefore, for an $N$-mode matrix, we obtain $2N$ imaginary eigenvalues~$\lambda_{j\pm} = \pm i \sigma_j$, which occur in conjugate pairs.  The $N$~symplectic eigenvalues are just the~$\sigma_j > 0$.  

We want the symplectic eigenvalues of the reduced density matrix, though, so we have to consider the covariance matrix truncated to the subset of $n$ modes.  Let~$\mat P$ be a matrix constructed in the following way: first, create a diagonal matrix that has a~1 in the diagonal entries corresponding to the nodes of the set to be kept and 0's everywhere else; then, remove the all-zero rows.  Also, let~$\bar {\mat P} = \left( \begin{smallmatrix} \mat P & \mat 0 \\ \mat 0 & \mat P \end{smallmatrix}\right)$, which is also not a square matrix.  
Here, $\mat P$ is an $n\times N$ matrix, and so $\bar {\mat P}$ is a $2n \times 2N$ matrix.

The covariance matrix of the reduced state of $n$ modes is then~$\tilde {\mat \Sigma} := \bar {\mat P} \mat \Sigma \bar {\mat P}^\tp$,
a $2n \times 2n$ matrix. Keep in mind that this truncated covariance matrix does not have a corresponding graph in our formalism, because it is generally not pure (if it were pure, then the original state would be trivially a product state).  
Since $\tilde {\mat \Omega} = \bar {\mat P} \mat \Omega \bar {\mat P}^\tp$ is the symplectic form for the truncated set of modes, we seek the ordinary eigenvalues of~$\bar {\mat P} \mat \Omega \mat \Sigma \bar {\mat P}^\tp$.

Our goal is to find a simple way to use the graph~$\mat Z$ to read off the bipartite entanglement entropy across an arbitrary boundary dividing the graph into the two subsets of nodes. In particular,  
we would like to find an interpretation of the symplectic eigenvalues, and hence the entanglement
in terms of the shape and the weights of the given graph.
For a general graph and arbitrary divisions into subsets of $n$ and $m$ modes, this interpretation
is not so straightforward. However, for specific graphs and bipartite splittings,
a simple, instructive connection between the entanglement and the graph can sometimes be found.

A particularly straightforward example is that of the canonical CV cluster states~\cite{Menicucci2006}.  Canonical CV cluster states are of the form~$\mat Z = i\epsilon \mat \id + \mat V$, where~$\epsilon = e^{-2r}$.
As shown in Figure~\ref{fig:complexgraphs}, the imaginary self-loops represent the initial squeezing of each node, while the matrix $\mat V$
contains the real weights between neighboring nodes (corresponding to the strength of the $\CZ$ gates,
which also contain squeezing~\cite{Braunstein2005}).
The covariance matrix for such a state is
\begin{align}
	\mat \Sigma &:= \frac 1 {2\epsilon}
	\begin{pmatrix}
		\mat \id		& \mat V	\\
		\mat V		& \epsilon^2 \mat \id + \mat V^2
	\end{pmatrix}
	\,.
\end{align}
This means that in order to obtain the ordinary eigenvalues
of $\bar {\mat P} \mat \Omega \mat \Sigma \bar {\mat P}^\tp$, we need to solve
\begin{align}
\label{eq:lambdaWV}
	0 &= \det
	\begin{pmatrix}
		\tilde {\mat V} - \lambda \mat \id	&	\epsilon^2 \mat \id + \tilde {\mat W}	\\
		-\mat \id					&	- \tilde {\mat V} - \lambda \mat \id
	\end{pmatrix}
	\nonumber \\
	&= \det ( \lambda^2 \mat \id - \tilde {\mat V}^2 + \epsilon^2 \mat \id + \tilde {\mat W} ) \,,
\end{align}
where~$\tilde {\mat V} = \mat P \mat V \mat P^\tp$ and~$\tilde {\mat W} = \mat P \mat V^2 \mat P^\tp$, and the second line follows because the bottom two blocks commute.  Notice that since $\mat \id - \mat P^\tp \mat P$ is a projector, $\tilde {\mat W} - \tilde{\mat V}^2 = \mat P \mat V (\mat \id - \mat P^\tp \mat P) \mat V \mat P^\tp \ge 0$, and thus we will label the eigenvalues of~$\tilde {\mat W} - \tilde{\mat V}^2 $ by~$\nu_j^2 \ge 0$.  Equation~\eqref{eq:lambdaWV} then gives~$\lambda_j^2 + \epsilon^2 = - \nu_j^2$.  The $2n$~eigenvalues~$\mu_{j\pm}$ of~$\tilde {\mat \Omega} \tilde {\mat \Sigma}$ are therefore given by~
\begin{align}
	\mu_{j\pm} = \pm \frac {\lambda_j} {2 \epsilon} = \pm \frac i 2 \sqrt {1 + \frac {\nu_j^2} {\epsilon^2}}\,,
\end{align}
which gives us the $n$ symplectic eigenvalues of~$\tilde {\mat \Sigma}$:
\begin{align}
	\sigma_{j} = \frac 1 2 \sqrt {1 + \frac {\nu_j^2} {\epsilon^2}} = \frac 1 2 \sqrt {1 + e^{+4r} \nu_j^2}\,.
\end{align}
The eigenvalues~$\nu_j^2$ of~$\tilde {\mat W} - \tilde {\mat V}^2$ contain the information
about the graph in question. In particular, 
when we consider the bipartite entanglement between any single node~$k$ of the graph
and the remaining nodes ($n=1$, $m=N-1$)
, we have
\begin{align}
	\tilde {\mat V} \mapsto \tilde V = \vphantom{\sum_l}V_{kk}\,, \qquad 
	\tilde {\mat W} \mapsto \tilde W = \sum_l V^2_{kl}\,.
\end{align}
As a result, we simply have to solve
$\lambda^2 + \epsilon^2 = - \sum_{\{ l \mid l\neq k\}} V^2_{kl}$, leading to a single symplectic
eigenvalue
\begin{align}
	\sigma = \frac 1 2 \sqrt {1 + e^{+4r} \sum_{\mathclap{\smash{\{ l \mid l\neq k\}}}} V^2_{kl}}\,.
\end{align}
In other words, besides the initial squeezing $r$,
the entanglement between any
node $k$ and the rest of the graph is determined by the number of its neighbors and
the strength of the link with each neighbor. Only those ``nodes'' of the graph
that have no neighbors at all would give $\sigma = \tfrac 1 2$, corresponding to a pure reduced state
and hence no entanglement with the actual nodes of the graph.
For any given $r$, however, any link with nonzero weight means $\sigma > \tfrac 1 2$
and thus entanglement. 

Quantitatively, both an increasing number of neighboring nodes
and an increasing strength of the links enhances the entanglement,
because the entropy grows with $\sigma$, and $\sigma - \tfrac 1 2$ represents the mean
excitation number of the reduced thermal state. For the special case
of equal unit weights $V_{kl} = 1$, 
we have $\sigma = \tfrac 1 2 \sqrt {1 + e^{+4r} N_k}$, where $N_k$ represents the number
of neighbors of node $k$. Thus, the maximum entanglement obtainable between a single
node and the rest of the graph is determined by the maximal number of neighboring
nodes, the so-called connectivity $C:= \max_k N_k$. For constant
connectivity, like in a 2D lattice with $C = 4$, the entropy will be bounded 
and does not increase with the size of the lattice. Only for an increasing
connectivity do we get larger entropies, and in principle the entanglement
between a single node and the rest of the graph may grow unboundedly with
the number of its links. This result is consistent with the bounds on the offline
squeezing per node needed to create a canonical CV cluster state of constant connectivity~\cite{Gu2009}.

\section{Conclusion}
\label{sec:conc}

We have generalized weighted graphs for continuous-variable~(CV) cluster states in a natural way to a graphical calculus for all Gaussian pure states.  The mathematics behind this generalization is not new~\cite{Simon1988,Arvind1995,Freitas1999,Bartlett2002}.  What \emph{is} new is interpreting the matrix formalism of Simon, Sudarshan, and Mukunda~\cite{Simon1988} as transformations on an undirected, complex-weighted graph.  This would be a mathematic triviality if it were not for the fact that applying this graphical interpretation to approximate CV cluster states and letting the overall squeezing go to infinity results in exactly the same graphs as are already being used to represent ideal CV cluster states~\cite{Zhang2006,Menicucci2006,vanLoock2007,Menicucci2007,Zaidi2008,Menicucci2008,Flammia2009,Gu2009}.  In addition, the graph transformation rules implied by the formalism immediately generalize \emph{all} of the existing graph transformation rules~\cite{Zhang2008,Zhang2008a,Zhang2010} to any Gaussian pure state and limit to these rules in the (unphysical) case of ideal CV cluster states.  It is these remarkable facts that make these results important.

This graphical formalism satisfies all five of the essential desired properties outlined in Section~\ref{subsec:gauss:props}.  We have also made headway on the three bullet points that followed.  Our achievements with the formalism thus far include using it to
\begin{itemize}
\item incorporate all details of finite squeezing within the CV cluster-state graph (Section~\ref{subsec:gauss:CVCS});
\item distinguish between different approximants to a given ideal CV cluster state at the graphical level (Section~\ref{subsec:gauss:CVCS});
\item incorporate $\cH$-graphs~\cite{Menicucci2007,Zaidi2008,Menicucci2008,Flammia2009} within the same graphical formalism through the exponential map (Section~\ref{subsec:gauss:Hgraph});
\item generalize the nullifier formalism~\cite{Gu2009} to all Gaussian pure states (Section~\ref{subec:gauss:nullifier});
\item use the nullifier formalism to define both an error matrix and a scalar approximation error for an approximate CV cluster state (Section~\ref{subec:gauss:nullifier});
\item define matrix and graphical transformation rules for a complete set of simple Gaussian unitary operations and quadrature measurements (Section~\ref{sec:transrules});
\item define the ``closest'' ideal CV cluster state to a given Gaussian pure state (Section~\ref{subsec:apps:closest});
\item analyze the optimality of the $\cH$-graph construction method with respect to this notion (Section~\ref{subsec:apps:Hopt}), including the specific examples of the two-mode squeezed state (Section~\ref{subsec:apps:TMS}) and a useful subclass of~$\cH$-graphs (Section~\ref{subsec:apps:selfinv});
\item demonstrate generation of a star-graph approximate CV cluster state from an approximate GHZ state made using an $\cH$-graph (Section~\ref{subsec:apps:GHZ});
\item identify classes of two-body Hamiltonians that have CV cluster states as their ground state (Section~\ref{subsec:apps:ground});
\item quantify bipartite entanglement in terms of the graphical formalism, with the explicit example of canonical CV cluster states (Section~\ref{subsec:apps:entanglement}).
\end{itemize}
We anticipate that this list will grow over time.  Specifically, we expect that the formalism will serve well when considering 
the propagation and manipulation of quantum information through approximate CV cluster states using homodyne detection 
and 
development of computer software for visualizing the effects of Gaussian operations on CV cluster states, 
as well as 
other uses not yet discovered.

We conclude with a few words about the prospects of fault tolerant quantum computing using Gaussian approximate CV cluster states.  Recent work~\cite{Ohliger2010} has demonstrated that all Gaussian approximants to ideal CV cluster states are inherently faulty when used for one-way quantum computing simply due to the fact that they are finitely squeezed.  This has led to suggestions (made in private communication) that these results forbid a fault-tolerant implementation of this quantum computing paradigm.  This is not the case.

That CV cluster states are error-prone has been known since the beginning~\cite{Menicucci2006}.  The main conclusion that should be drawn from Reference~\olcite{Ohliger2010} is that there is no ``magic pill'' to eliminate the effects of finite squeezing using a simple qubit (or other) encoding scheme.  Instead, fault tolerance must be addressed from the very beginning because unlike qubit cluster states, which remain physical in the limit of zero errors in preparation and use, CV cluster states are unphysical in this limit, since ideal states require infinite energy.  This is not in any way a show-stopper, however, and the authors of Reference~\olcite{Ohliger2010} go to great lengths to show several possible approaches to error correction that do not fall victim to their no-go theorems.  (It would be interesting to apply the matrix-product-state calculations from that paper to noisy qubit cluster states to see if analogous results are obtained in that context, in order to ensure we are comparing apples with apples.)  Fault tolerance in CV one-way quantum computing remains an important open problem and an active area of research.  

\begin{acknowledgments}
We are grateful for discussions with and assistance from Daniel Gottesman, Rolando Somma, Jon Yard, Yeong-Cherng Liang, Andrew Doherty, and Stephen Bartlett.  P.v.L.~acknowledges the DFG for financial support through the Emmy Noether programme.  Research at Perimeter Institute is supported by the Government of Canada through Industry Canada and by the Province of Ontario through the Ministry of Research \& Innovation.
\end{acknowledgments}


\appendix

\section{Derivation of the complex-weighted graph representation}
\label{app:derivations}

Here we derive the connection between a Gaussian pure state and its complex-weighted graph~$\mat Z$.  There is a vast mathematical literature on the set of all allowable graphs~$\mat Z$, known as the \emph{Siegel upper half-space}.  The main results for our purposes can be found in Reference~\olcite{Arvind1995}, with more details in Reference~\olcite{Simon1988}.  An extensive review of the Siegel upper half-space and its connection to the symplectic group can be found in the Ph.D. thesis by Freitas~\cite{Freitas1999} and the references therein.


\begin{definition}
The graph corresponding to a Gaussian pure state is~$\mat Z := \avgg {\opvec q \opvec q^\tp }^{-1} \avgg { \opvec q \opvec p^\tp }$.
\end{definition}
This graph was defined in Eq.~\eqref{eq:Zfromsigma}.  We wish to prove several properties of all such graphs.

\begin{theorem}
Every graph corresponding to a Gaussian pure state is complex-weighted, undirected, and unique and has positive definite imaginary part.
\end{theorem}

\begin{proof}
Let~$\mat Z$ be a graph corresponding to a Gaussian pure state.  That $\mat Z$~is complex weighted and unique is guaranteed by the definition: expectation values of operator-valued matrices result in matrices of complex numbers, corresponding to a complex-weighted graph.  It is unique for a given state because expectation values are uniquely determined by the state.

An undirected graph has a symmetric adjacency matrix.  For showing the symmetry of~$\mat Z$, the second line of Eq.~\eqref{eq:Zfromsigma} is useful, and we repeat it here for reference:
\begin{align}
\label{eq:Zfromsigmaline2}
	\mat Z &= \frac 1 2 \avgg { \opvec q \opvec q^\tp }^{-1} \left( i \mat \id + \avgg { \{ \opvec q, \opvec p^\tp \} } \right)\,.
\end{align}
To prove the symmetry of~$\mat Z$ we refer to Eq.~\eqref{eq:covfromS}, which shows that any Gaussian pure state has a covariance matrix that is $\frac 1 2$ times a symplectic matrix:
\begin{align}
	2\mat \Sigma = \mat S \mat S^\tp \in \Sp(2N,\reals)\,,
\end{align}
since it is the product of a symplectic matrix~$\mat S$ and its transpose.  While $\mat S$~is not unique for a given Gaussian pure state, $\mat \Sigma$~is, and we partition it as follows:
\begin{align}
	\mat \Sigma =
	\frac 1 2
	\begin{pmatrix}
		\mat A	& \mat B	\\
		\mat B^\tp	& \mat D
	\end{pmatrix}
	\,.
\end{align}
Being a covariance matrix, $\mat \Sigma = \mat \Sigma^\tp > 0$.  (Being the covariance matrix for a valid quantum state requires more than this~\cite{Braunstein2005a}, but we don't need it.)  This implies immediately that $\mat A = \mat A^\tp > 0$, and $\mat D = \mat D^\tp > 0$.  That $2\mat \Sigma$~is symplectic requires, additionally, that
\begin{align}
	\begin{pmatrix}
		\mat 0	& \mat \id	\\
		-\mat \id	& \mat 0
	\end{pmatrix}
	= \mat \Omega &= 4\mat \Sigma \mat \Omega \mat \Sigma \nonumber \\
	&=
	\begin{pmatrix}
		\mat A	& \mat B	\\
		\mat B^\tp	& \mat D
	\end{pmatrix}
	\begin{pmatrix}
		\mat 0	& \mat \id	\\
		-\mat \id	& \mat 0
	\end{pmatrix}
	\begin{pmatrix}
		\mat A	& \mat B	\\
		\mat B^\tp	& \mat D
	\end{pmatrix}
	\nonumber \\
	&=
	\begin{pmatrix}
		\mat A \mat B^\tp - \mat B \mat A	& \mat A \mat D - \mat B^2	\\
		(\mat B^\tp)^2 - \mat D \mat A		& \mat B^\tp \mat D - \mat D \mat B
	\end{pmatrix}
	\,.
\end{align}
From this we can see immediately that
\begin{subequations}
\label{eq:Sigmasymplectic}
\begin{align}
\label{eq:BAsymm}
	(\mat B \mat A)^\tp &= \mat B \mat A\,, \\
\label{eq:DBsymm}
	(\mat D \mat B)^\tp &= \mat D \mat B\,, \\
\label{eq:ADdef}
	\mat A \mat D &= \mat \id + \mat B^2\,.
\end{align}
\end{subequations}
From the definition of the covariance matrix, we have that $\avgg {\opvec q \opvec q^\tp} = \mat A$, and $\frac 1 2 \avgg { \{\opvec q, \opvec p^\tp \} } = \mat B$.  That~$\mat A = \mat A^\tp > 0$ guarantees immediately that $\Im \mat Z = \mat A^{-1}$ exists and is also symmetric and positive definite.  The only remaining item to prove is that $\Re \mat Z = \mat A^{-1} \mat B$ is symmetric.  Equation~\eqref{eq:BAsymm} gives
\begin{align}
\label{eq:ReZsymm}
	&&(\mat B \mat A)^\tp &= \mat B \mat A \nonumber \\
	&\Longrightarrow & \mat A^{-1} (\mat B \mat A)^\tp \mat A^{-1} &= \mat A^{-1} \mat B \mat A \mat A^{-1} \nonumber \\
	&\Longrightarrow & (\mat A^{-1} \mat B)^\tp &= \mat A^{-1} \mat B\,,
\end{align}
since $\mat A = \mat A^\tp$.  Therefore, $\mat Z$~is an undirected graph.
\end{proof}

\begin{theorem}
Every complex-weighted, undirected graph with positive-definite imaginary part represents a unique Gaussian pure state (up to an arbitrary overall phase).
\end{theorem}

\begin{proof}
Let~$\mat Z$ be the graph in question.  We'll split it up into its real and imaginary parts as usual: $\mat Z = i\mat U + \mat V$, where~$\mat U = \mat U^\tp > 0$, and~$\mat V = \mat V^\tp$.  To find the Gaussian pure state that $\mat Z$~represents, we will construct a covariance matrix from it and prove that it satisfies the conditions to be a valid covariance matrix for a Gaussian pure state.

We now define
\begin{align}
\label{eq:SigmaZ}
	\mat \Sigma_{\mat Z} &:= \frac 1 2
	\begin{pmatrix}
		\mat U^{-1}			& \mat U^{-1} \mat V	\\
		\mat V \mat U^{-1}		& \mat U + \mat V \mat U^{-1} \mat V
	\end{pmatrix}
	\,,
\end{align}
just as in Eq.~\eqref{eq:covmatuv}.  In this case, though, we need to prove that it is a valid covariance matrix for a Gaussian pure state given the assumptions made about~$\mat Z$.  A straightforward way to do this is to show that it is the result of conjugation of the ground state covariance matrix~$\frac 1 2 \mat \id$ by a symplectic matrix (which necessarily represents a Gaussian unitary operation).  We define
\begin{align}
	\mat S_{\mat Z} :=
	\begin{pmatrix}
		\mat U^{-1/2}			& \mat 0	\\
		\mat V \mat U^{-1/2}		& \mat U^{1/2}
	\end{pmatrix}
	\,,
\end{align}
paralleling Eq.~\eqref{eq:Suv}.  One can verify directly that $\mat S_{\mat Z} \mat \Omega \mat S_{\mat Z}^\tp = \mat \Omega$, guaranteeing the symplectic nature of~$\mat S_{\mat Z}$.  Then we have
\begin{align}
	\mat \Sigma_{\mat Z} = \frac 1 2 \mat S_{\mat Z} \mat S_{\mat Z}^\tp\,,
\end{align}
which shows that $\mat Z$~represents a valid Gaussian pure state.

To show uniqueness of the state represented (up to overall phase), we assume that there exists another covariance matrix~$\tilde {\mat \Sigma}_{\mat Z} \neq \mat \Sigma_{\mat Z}$ that is represented by~$\mat Z$.  Recalling Eq.~\eqref{eq:ADdef}, which holds for an arbitrary pure-Gaussian-state covariance matrix, we write the blocks of~$\tilde {\mat \Sigma}_{\mat Z}$ as
\begin{align}
\label{eq:tildeSigma}
	\tilde {\mat \Sigma}_{\mat Z} &= \frac 1 2
	\begin{pmatrix}
		\mat A	& \mat B	\\
		\mat B^\tp	& \mat A^{-1} + \mat A^{-1} \mat B^2
	\end{pmatrix}
	\,.
\end{align}
Using Eq.~\eqref{eq:Zfromsigmaline2}, the graph representing this state is~$\mat A^{-1} (i \mat \id + \mat B)$.  By assumption, this must be the same as~$\mat Z = i\mat U + \mat V$, which gives~$\mat A = \mat U^{-1}$ and~$\mat B = \mat U^{-1} \mat V$.  Plugging these back into Eq.~\eqref{eq:tildeSigma} shows that $\tilde {\mat \Sigma}_{\mat Z} = \mat \Sigma_{\mat Z}$, in contradiction with the assumption.  Therefore, $\mat \Sigma_{\mat Z}$~is unique for a given graph~$\mat Z$.
\end{proof}


\begin{theorem}
The transformation law for graphs representing Gaussian pure states under Gaussian unitary operations is given by Eq.~\eqref{eq:Zprime}.
\end{theorem}

\begin{proof}
Since we have a unique way of passing from the covariance matrix for a Gaussian pure state to its graph and back again, our method will be to show the action of an arbitrary symplectic operation on the covariance matrix and then extract the new graph from it.  Rather than dealing with the covariance matrix alone, it will be useful instead to consider the combination
\begin{align}
\label{eq:SigmaOmega}
	\mat \Sigma_{\mat Z} - \frac i 2 \mat \Omega &= \frac 1 2
	\begin{pmatrix}
		\mat U^{-1}					& \mat U^{-1} \mat V - i \mat \id	\\
		\mat V \mat U^{-1} + i \mat \id		& \mat U + \mat V \mat U^{-1} \mat V
	\end{pmatrix}
	\nonumber \\
	&= \frac 1 2
	\begin{pmatrix}
		\mat \id	\\
		\mat Z
	\end{pmatrix}
	\mat U^{-1}
	\begin{pmatrix}
		\mat \id	&	\mat Z^*
	\end{pmatrix}
	\nonumber \\
	&= \frac 1 2
	\begin{pmatrix}
		\mat \id	\\
		\mat Z
	\end{pmatrix}
	\left[ \frac 1 {2i} (\mat Z - \mat Z^*) \right]^{-1}
	\begin{pmatrix}
		\mat \id	&	\mat Z^*
	\end{pmatrix}
	\nonumber \\
	&= i
	\begin{pmatrix}
		\mat \id	\\
		\mat Z
	\end{pmatrix}
	\left[
	\begin{pmatrix}
		\mat \id	&	\mat Z^*
	\end{pmatrix}
	\mat \Omega
	\begin{pmatrix}
		\mat \id	\\
		\mat Z
	\end{pmatrix}
	\right]^{-1}
	\begin{pmatrix}
		\mat \id	&	\mat Z^*
	\end{pmatrix}
	\,,
\end{align}
and similarly for~$\mat \Sigma_{\mat Z'} - \frac i 2 \mat \Omega$, with~$\mat Z \mapsto \mat Z'$.  Equation~\eqref{eq:Zprime} requires that $\mat \Sigma_{\mat Z'} = \mat S \mat \Sigma_{\mat Z} \mat S^\tp$, with
\begin{align}
\label{eq:SABCDrepeated}
	\mat S =
	\begin{pmatrix}
		\mat A	& \mat B	\\
		\mat C	& \mat D
	\end{pmatrix}
\end{align}
from Eq.~\eqref{eq:SABCD}, repeated here for reference.  Using the symplectic property of~$\mat S$ allows us to write
\begin{align}
\label{eq:SigmaOmegaprime}
	&\mat \Sigma_{\mat Z'} - \frac i 2 \mat \Omega \nonumber \\
	&= \mat S \left(\mat \Sigma_{\mat Z} - \frac i 2 \mat \Omega \right) \mat S^\tp \nonumber \\
	&= i \mat S
	\begin{pmatrix}
		\mat \id	\\
		\mat Z
	\end{pmatrix}
	\left[
	\begin{pmatrix}
		\mat \id	&	\mat Z^*
	\end{pmatrix}
	\mat \Omega
	\begin{pmatrix}
		\mat \id	\\
		\mat Z
	\end{pmatrix}
	\right]^{-1}
	\begin{pmatrix}
		\mat \id	&	\mat Z^*
	\end{pmatrix}
	\mat S^\tp
	\nonumber \\
	&= i \mat S
	\begin{pmatrix}
		\mat \id	\\
		\mat Z
	\end{pmatrix}
	\left[
	\begin{pmatrix}
		\mat \id	&	\mat Z^*
	\end{pmatrix}
	\mat S^\tp \mat \Omega \mat S
	\begin{pmatrix}
		\mat \id	\\
		\mat Z
	\end{pmatrix}
	\right]^{-1}
	\begin{pmatrix}
		\mat \id	&	\mat Z^*
	\end{pmatrix}
	\mat S^\tp
	\,.
\end{align}
Notice that
\begin{align}
\label{eq:SonIZ}
	\mat S
	\begin{pmatrix}
		\mat \id	\\
		\mat Z
	\end{pmatrix}
	&=
	\begin{pmatrix}
		\mat A + \mat B \mat Z	\\
		\mat C + \mat D \mat Z
	\end{pmatrix}
	=
	\begin{pmatrix}
		\mat \id	\\
		\tilde {\mat Z}
	\end{pmatrix}
	(\mat A + \mat B \mat Z)\,,
\end{align}
where $\tilde {\mat Z} = (\mat C + \mat D \mat Z)(\mat A + \mat B \mat Z)^{-1}$.  The fact that~$\Im \mat Z > 0$ guarantees that the inverse exists (unless $\mat B = \mat 0$, in which case invertibility of~$\mat S$ guarantees that $\mat A^{-1}$~exists).  Equation~\eqref{eq:Zprime} then amounts to the claim that $\mat Z' = \tilde {\mat Z}$ and also that $\tilde {\mat Z}$~satisfies all the requirements for representing a Gaussian pure state (i.e.,\ symmetry and positive-definite imaginary part).

To show the symmetry of $\tilde {\mat Z}$, we use a trick similar to that used in Eq.~\eqref{eq:SigmaOmegaprime}:
\begin{align}
\label{eq:Ztildesymm}
	\mat 0 &= \mat Z - \mat Z^\tp \nonumber \\
	&=
	\begin{pmatrix}
		\mat \id	&	\mat Z^\tp
	\end{pmatrix}
	\mat \Omega
	\begin{pmatrix}
		\mat \id	\\
		\mat Z
	\end{pmatrix}
	\nonumber \\
	&=
	\begin{pmatrix}
		\mat \id	&	\mat Z^\tp
	\end{pmatrix}
	\mat S^\tp \mat \Omega \mat S
	\begin{pmatrix}
		\mat \id	\\
		\mat Z
	\end{pmatrix}
	\nonumber \\
	&= (\mat A + \mat B \mat Z)^\tp (\tilde {\mat Z} - \tilde {\mat Z}^\tp) (\mat A + \mat B \mat Z)\,.
\end{align}
Since $(\mat A + \mat B \mat Z)$~is invertible, $\tilde {\mat Z} = \tilde {\mat Z}^\tp$.  Similarly, to show that $\Im \tilde {\mat Z} > 0$, we calculate
\begin{align}
\label{eq:ImZtilde}
	\Im \mat Z &= \frac {1} {2i} (\mat Z - \mat Z^*) \nonumber \\
	&= \frac {1} {2i}
	\begin{pmatrix}
		\mat \id	&	\mat Z^*
	\end{pmatrix}
	\mat \Omega
	\begin{pmatrix}
		\mat \id	\\
		\mat Z
	\end{pmatrix}
	\nonumber \\
	&= \frac {1} {2i}
	\begin{pmatrix}
		\mat \id	&	\mat Z^*
	\end{pmatrix}
	\mat S^\tp \mat \Omega \mat S
	\begin{pmatrix}
		\mat \id	\\
		\mat Z
	\end{pmatrix}
	\nonumber \\
	&= \frac {1} {2i} (\mat A + \mat B \mat Z)^\herm (\tilde {\mat Z} - \tilde {\mat Z}^*) (\mat A + \mat B \mat Z) \nonumber \\
	&= (\mat A + \mat B \mat Z)^\herm (\Im \tilde {\mat Z}) (\mat A + \mat B \mat Z)\,.
\end{align}
Inverting this relation verifies that~$(\Im \mat Z > 0) \Longrightarrow (\Im \tilde {\mat Z} > 0)$.

We have shown that~$\tilde {\mat Z}$ satisfies the requirements for representing a Gaussian pure state.  All that's left to show is that $\mat Z' = \tilde {\mat Z}$.  To do this, we first take the conjugate transpose of Eq.~\eqref{eq:SonIZ}, giving
\begin{align}
\label{eq:IZstaronS}
	\begin{pmatrix}
		\mat \id	&	\mat Z^*
	\end{pmatrix}
	\mat S^\tp
	&=
	(\mat A^\tp + \mat Z^* \mat B^\tp)
	\begin{pmatrix}
		\mat \id	&	\tilde {\mat Z}^*
	\end{pmatrix}
	\,.
\end{align}
Plugging Eqs.~\eqref{eq:SonIZ} and~\eqref{eq:IZstaronS} into Eq.~\eqref{eq:SigmaOmegaprime} and canceling the appropriate factors gives
\begin{align}
\label{eq:SigmaOmegatilde}
	\mat \Sigma_{\mat Z'} - \frac i 2 \mat \Omega &= i
	\begin{pmatrix}
		\mat \id	\\
		\tilde {\mat Z}
	\end{pmatrix}
	\left[
	\begin{pmatrix}
		\mat \id	&	\tilde {\mat Z}^*
	\end{pmatrix}
	\mat \Omega
	\begin{pmatrix}
		\mat \id	\\
		\tilde {\mat Z}
	\end{pmatrix}
	\right]^{-1}
	\begin{pmatrix}
		\mat \id	&	\tilde {\mat Z}^*
	\end{pmatrix}
	\,.
\end{align}
Clearly, $\tilde {\mat Z}$~appears everywhere that $\mat Z'$~should appear.  The reader can check that solving this equation for $\mat \Sigma_{\mat Z'}$ and extracting its graph~$\mat Z' = \avgg {\opvec q \opvec q^\tp }^{-1} \avgg { \opvec q \opvec p^\tp }$ does, in fact, show that $\mat Z' = \tilde {\mat Z}$.  Since the graph for a Guassian state is unique, this verifies Eq.~\eqref{eq:Zprime}.  This transformation law is called a \emph{generalized M\"obius transformation}, and the interested reader is directed to Reference~\olcite{Freitas1999} for a more in-depth mathematical analysis.
\end{proof}

\section{Derivation of the closest CV cluster state to a given Gaussian pure state}
\label{app:derivationclosest}

There are other matrix models for Gaussian pure states besides the one we are using here.  One of these, based on the \emph{Siegel disc}~\cite{Freitas1999}, is useful for these calculations.  (We will forgo presentation of the entire model, referring the interested reader to Reference~\olcite{Freitas1999} and the references therein.)  Based on this model, we define
\begin{align}
\label{eq:Kdef}
	\mat K &:= (\mat \id + i \mat Z) (\mat \id - i \mat Z)^{-1} \nonumber \\
	&=
	\begin{pmatrix}
		\mat \id	& i\mat \id
	\end{pmatrix}
	\begin{pmatrix}
		\mat \id	\\
		\mat Z
	\end{pmatrix}
	\left[
	\begin{pmatrix}
		\mat \id	& -i\mat \id
	\end{pmatrix}
	\begin{pmatrix}
		\mat \id	\\
		\mat Z
	\end{pmatrix}
	\right]^{-1}
	\,.
\end{align}
We will also need the following:
\begin{align}
\label{eq:KZidentity}
	\mat \id + \mat K &= [(\mat \id - i \mat Z) + (\mat \id + i \mat Z)] (\mat \id - i \mat Z)^{-1} \nonumber \\
	&= 2 (\mat \id - i \mat Z)^{-1} \,,
\end{align}
as well as
\begin{align}
\label{eq:ZKidentity}
	\mat \id - i \mat Z = 2 (\mat \id + \mat K)^{-1}\,.
\end{align}
Notice that since $\mat \id + \mat K$ is symmetric, so is~$\mat K$.  We repeat Eq.~\eqref{eq:Smattheta} here for reference:
\begin{align}
	\mat S_{\mat \theta} =
	\begin{pmatrix}
		\cos \mat \theta	&	\sin \mat \theta	\\
		-\sin \mat \theta	&	\cos \mat \theta		
	\end{pmatrix}
	\,,
\end{align}
where $\mat \theta = \diag (\theta_1, \dotsc, \theta_N)$.  We can derive the transformation law for~$\mat K$ with respect to these phase shifts [Cf.~Eq.~\eqref{eq:SonIZ}]:
\begin{align}
	\mat K' &=
	\begin{pmatrix}
		\mat \id	& i\mat \id
	\end{pmatrix}
	\begin{pmatrix}
		\mat \id	\\
		\mat Z'
	\end{pmatrix}
	\left[
	\begin{pmatrix}
		\mat \id	& -i\mat \id
	\end{pmatrix}
	\begin{pmatrix}
		\mat \id	\\
		\mat Z'
	\end{pmatrix}
	\right]^{-1}
	\nonumber \\
	&=
	\begin{pmatrix}
		\mat \id	& i\mat \id
	\end{pmatrix}
	\mat S_{\mat \theta}
	\begin{pmatrix}
		\mat \id	\\
		\mat Z
	\end{pmatrix}
	\left[
	\begin{pmatrix}
		\mat \id	& -i\mat \id
	\end{pmatrix}
	\mat S_{\mat \theta}
	\begin{pmatrix}
		\mat \id	\\
		\mat Z
	\end{pmatrix}
	\right]^{-1}
	\nonumber \\
	&=
	e^{-i \mat \theta}
	\begin{pmatrix}
		\mat \id	& i\mat \id
	\end{pmatrix}
	\begin{pmatrix}
		\mat \id	\\
		\mat Z
	\end{pmatrix}
	\left[
	e^{i \mat \theta}
	\begin{pmatrix}
		\mat \id	& -i\mat \id
	\end{pmatrix}
	\begin{pmatrix}
		\mat \id	\\
		\mat Z
	\end{pmatrix}
	\right]^{-1}
	\nonumber \\
	&= e^{-i \mat \theta} \mat K e^{-i \mat \theta}\,.
\end{align}

An extremum of~$\tfrac 1 2 \tr \mat U'$ occurs when~$\partial_j \tr \mat U' = \vec 0$, where~$\partial_j := \tfrac {\partial} {\partial \theta_j}$.  Let's calculate the left-hand side:
\begin{align}
	\partial_j \tr \mat U' &= \partial_j \tr (\mat \id + \mat U') \nonumber \\
	&= \tfrac 1 2 \partial_j \tr (\mat \id - i \mat Z') + \cc \nonumber \\
	&= \partial_j \tr [ (\mat \id + \mat K')^{-1} ] + \cc \nonumber \\
	&= \tr \{ -(\mat \id + \mat K')^{-1} [ \partial_j (\mat \id + \mat K') ] (\mat \id + \mat K')^{-1} \} + \cc \nonumber \\
	&= \tr \{ -(\mat \id + \mat K')^{-1} (-i)e^{-i \mat \theta_j} \mat K e^{-i \mat \theta} (\mat \id + \mat K')^{-1} \} \nonumber \\
	&\qquad + \tr \{ \text{transpose} \} + \cc \nonumber \\
	&= 2i \tr \{ (\mat \id + \mat K')^{-1} e^{-i \mat \theta_j} \mat K e^{-i \mat \theta} (\mat \id + \mat K')^{-1} \} + \cc \nonumber \\
	&= 2i [\mat K' (\mat \id + \mat K')^{-2} ]_{jj} + \cc\,,
\end{align}
where ``$\cc$''~stands for ``complex conjugate,'' and all entries of $\mat \theta_j$~are zero except for the $(j,j)^\text{th}$, which equals~$\theta_j$.  We will also need the following:
\begin{align}
	\mat K' (\mat \id + \mat K')^{-2} &= (\mat \id + \mat K')^{-1} [\mat \id - (\mat \id + \mat K')^{-1}] \nonumber \\
	&= (\mat \id + \mat K')^{-1} - (\mat \id + \mat K')^{-2} \nonumber \\
	&= \frac 1 2 (\mat \id - i\mat Z') - \frac 1 4 (\mat \id - i\mat Z')^2 \nonumber \\
	&= \frac 1 4 (\mat \id + \mat Z'^2)\,.
\end{align}
Then
\begin{align}
	\partial_j \tr \mat U' &= \frac i 2 (\mat \id + \mat Z'^2)_{jj} + \cc \nonumber \\
	&= -\Im (\mat Z'^2)_{jj} \nonumber \\
	&= -( \mat U' \mat V' )_{jj}\,.
\end{align}
Setting this to~0 for an extremum verifies Eq.~\eqref{eq:extremumcond}.

To show that we have a minimum rather than just an extremum, the Hessian of $\tr \mat U'$ must be positive definite.  The Hessian matrix has entries
\begin{align}
\label{eq:dkdjtrU1}
	\partial_k \partial_j \tr \mat U' 
	&= \frac i 2 \partial_k (\mat \id + \mat Z'^2)_{jj} + \cc \nonumber \\
	&= \frac i 2 \partial_k [(\mat \id + i\mat Z') (\mat \id - i\mat Z')]_{jj} + \cc
\end{align}
In addition to Eq.~\eqref{eq:ZKidentity}, we can find another similar relation by inverting Eq.~\eqref{eq:Kdef}:
\begin{align}
\label{eq:KinvZidentity}
	\mat \id + \mat K^{-1} &= [(\mat \id + i \mat Z) + (\mat \id - i \mat Z)] (\mat \id + i \mat Z)^{-1} \nonumber \\
	&= 2 (\mat \id + i \mat Z)^{-1} \,,
\end{align}
and thus, also,
\begin{align}
\label{eq:ZKinvidentity}
	\mat \id + i \mat Z = 2(\mat \id + \mat K^{-1})^{-1}\,,
\end{align}
and similarly for the primed matrices.  We now define the placeholder matrix
\begin{align}
\label{eq:Qdef}
	\mat Q' &:= (\mat \id + \mat K'^{-1}) (\mat \id + \mat K') \nonumber \\
	&= 2\mat \id + \mat K'^{-1} + \mat K'\,.
\end{align}
Notice that
\begin{align}
\label{eq:Qinv}
	\mat \id + \mat Z'^2 = 4\mat Q'^{-1}\,.
\end{align}
The partial derivatives of~$\mat Q'$ are given by
\begin{align}
\label{eq:dQ}
	\partial_k \mat Q' &= \partial_k \mat K' + \partial_k \mat K'^{-1} \nonumber \\
	&= \partial_k \mat K' - \mat K'^{-1} (\partial_k \mat K') \mat K'^{-1} \nonumber \\
	&= -i\mat \pi_k ( \mat K' - \mat K'^{-1}) - i ( \mat K' - \mat K'^{-1}) \mat \pi_k\,,
\end{align}
where the $(k,k)^\text{th}$ entry of~$\mat \pi_k$ equals~1, while all others are~0.  We can plug these results into Eq.~\eqref{eq:dkdjtrU1}:
\begin{align}
\label{eq:dkdjtrU2}
	&\quad \partial_k \partial_j \tr \mat U' \nonumber \\
	&= 2i \tr [\mat \pi_j \partial_k \mat Q'^{-1}] + \cc \nonumber \\
	&= -2i \tr [\mat \pi_j \mat Q'^{-1} (\partial_k \mat Q') \mat Q'^{-1} ] + \cc \nonumber \\
	&= -4 \tr [\mat \pi_j \mat Q'^{-1} \mat \pi_k ( \mat K' - \mat K'^{-1}) \mat Q'^{-1} ] + \cc \nonumber \\
	&= -4 \tr \{\mat \pi_j \mat Q'^{-1} \mat \pi_k [ (\mat \id + \mat K') \mat Q'^{-1} \nonumber \\
	&\qquad \qquad - (\mat \id + \mat K'^{-1})\mat Q'^{-1} ] \} + \cc \nonumber \\
	&= -4 \tr \{\mat \pi_j \mat Q'^{-1} \mat \pi_k [ (\mat \id + \mat K'^{-1})^{-1} - (\mat \id + \mat K')^{-1}] \} + \cc \nonumber \\
	&= -2 \tr \{\mat \pi_j \mat Q'^{-1} \mat \pi_k [ (\mat \id + i\mat Z') - (\mat \id - i\mat Z')] \} + \cc \nonumber \\
	&= -4i \tr (\mat \pi_j \mat Q'^{-1} \mat \pi_k \mat Z') + \cc \nonumber \\
	&= 2 \Im \tr [\mat \pi_j (\mat \id + \mat Z'^2) \mat \pi_k \mat Z'] \nonumber \\
	&= 2 \Im [(\mat \id + \mat Z'^2)_{jk} \mat Z'_{kj}] \nonumber \\
	&= 2 \Im [(\mat \id + \mat Z'^2) \circ \mat Z']_{jk}\,,
\end{align}
where $\circ$~represents the Hadamard (entrywise) product of two matrices.  Requiring the matrix with these entries to be positive definite verifies Eq.~\eqref{eq:Hesspos}.

\section{Stabilizers for Gaussian pure states}
\label{app:stabilizers}


To find the stabilizer operators for the finitely squeezed,
canonical CV cluster states,
we start by constructing the stabilizer of the vacuum state~$\ket 0$ of
a qumode~$\op a_k$~\cite{Bartlett2002}. For the dimensionless
quadrature operators~$\op q$ and~$\op p$, where $\op a = \tfrac {1} {\sqrt 2} (\op q + i \op p)$,
we obtain
\begin{equation}
	\ket{0} = \exp(\alpha\op a_k) \ket{0} = \exp\left[ \tfrac {\alpha} {\sqrt 2} (\op q_k + i \op p_k)\right] \ket{0}\,.
\end{equation}
Further, from this we need the stabilizer for a single-mode squeezed state~$\op S(r_k) \ket 0$,
with a squeezing parameter $r_k > 0$ and
a $\op q$-squeezing operator~$\op S(r_k) = \exp[\tfrac i 2 r_k(\op q_k \op p_k + \op p_k \op q_k)]$.
The stabilizer equation can be written as
\begin{align}
	\op S(r_k) \ket 0&= \op S(r_k) \exp\left[\tfrac {\alpha} {\sqrt 2} (\op q_k + i \op p_k) \right] \op S^\dagger(r_k) \op S(r_k) \ket 0 \nonumber \\
	&= \exp\left[\tfrac {\alpha} {\sqrt 2} (e^{+r_k}\op q_k + i e^{-r_k}\op p_k) \right] \op S(r_k) \ket 0\,.
\end{align}
In the case of momentum squeezing, with $\op S(-r_k)$, we have
\begin{align}
	\op S(-r_k) \ket 0 &= \op S(-r_k) \exp\left[\tfrac {\alpha} {\sqrt 2} (\op q_k + i \op p_k)\right] \op S^\dagger(-r_k) \nonumber \\
	&\qquad \times \op S(-r_k) \ket 0 \nonumber \\
	&= \exp\left[\tfrac {\alpha} {\sqrt 2} (e^{-r_k}\op q_k + i e^{+r_k}\op p_k)\right] \op S(-r_k) \ket 0 \,.
\end{align}
Let us rewrite this as
\begin{equation}
	\exp \left(-\tfrac 1 4 \alpha^2 \right) \op X_k\left(-\tfrac {\alpha} {\sqrt 2} e^{+r_k}\right)\op Z_k\left(-i\tfrac {\alpha} {\sqrt 2} e^{-r_k} \right)\,,
\end{equation}
formally using the WH shift operators $\op X(s) = e^{- i s \op p}$ and $\op Z(s) = e^{i s \op q}$.
Now we define $\alpha := - \sqrt{2} e^{-r_k} s$ such that
the momentum-squeezed stabilizer becomes
\begin{equation}
\label{eq:squeezedstab}
	\exp \left(-\tfrac 1 2 e^{-2 r_k} s^2 \right) \op X_k(s)\op Z_k(i e^{-2 r_k} s)\,.
\end{equation}
In the limit of infinite $p$-squeezing $r_k\to\infty$, this operator
approaches $\op X_k(s)$, which stabilizes the zero-eigenstate~$\ketsub{0}{p_k}$, with
$\op X_k(s) \ketsub{0}{p_k} = \ketsub{0}{p_k}$ for all~$s\in \mathbb{R}$, as expected.

Now we can proceed to create CV cluster states in the canonical way: by pairwise applying the $\CZ$~gates, indicated as~$\opCZ^{kl}$ for a link between nodes~$k$ and~$l$.  The $N$~stabilizers of the initial $N$~momentum-squeezed modes, showing in Eq.~\eqref{eq:squeezedstab}, with~$k=1,2,\dotsc,N$, are then transformed for each interaction with neighbor~$l$ as
\begin{multline}
	\exp \left(-\tfrac 1 2 e^{-2 r_k} s^2 \right)\opCZ^{kl} \op X_k(s)\opCZ^{kl \dagger}\opCZ^{kl}\op Z_k(i e^{-2 r_k} s)\opCZ^{kl \dagger} 
	\\
	= \exp \left(-\tfrac 1 2 e^{-2 r_k} s^2 \right) \op X_k(s)\op Z_l(s)\op Z_k(i e^{-2 r_k} s)\,.
\end{multline}
Eventually, collecting all these interactions, we obtain the $N$
new stabilizers
\begin{eqnarray}
	\exp \left(-\tfrac 1 2 e^{-2 r_k} s^2 \right) \op X_k(s)\op Z_k(i e^{-2 r_k} s)\prod_{\mathclap{l\in \mathcal N(k)}}\op Z_l(s)\,,
\end{eqnarray}
where $\mathcal N(k)$~is the set of neighbors of~$k$.  In the limit of infinite squeezing ($r_k\to\infty$),
we get back the well-known, ideal CV cluster-state stabilizers for unweighted graphs.
However, this time, the above stabilizers also do the job
for finite squeezing and uniquely represent the corresponding
approximate cluster state. The nullifiers are obtained by taking the log of the stabilizers:
\begin{align}
	\exp \bigl(-\tfrac 1 2 &e^{-2 r_k} s^2 \bigr) \op X_k(s)\op Z_k(i e^{-2 r_k} s)\prod_{\mathclap{l\in \mathcal N(k)}}\op Z_l(s) \nonumber\\
	& = \exp \left(-\tfrac 1 2 e^{-2 r_k} s^2 \right) \exp\left[-is \left(\op p_k - i e^{-2 r_k} \op q_k \right)\right] \nonumber \\
	&\qquad \times \exp \left(+\tfrac 1 2 e^{-2 r_k} s^2 \right)\prod_{\mathclap{l\in \mathcal N(k)}}\exp(is \op q_l)\nonumber\\
	&= \exp\left[ -is \Bigl(\op p_k - i e^{-2 r_k} \op q_k - \sum_{l} \op q_l \Bigr)\right]\,,
\end{align}
for all~$k=1,2, \dotsc, N$ and for all~$s\in \mathbb{R}$.  The nullifiers are therefore
\begin{equation}
	\op p_k - i e^{-2 r_k} \op q_k - \sum_{\mathclap{l\in \mathcal N(k)}} \op q_l \quad \forall k\,.
\end{equation}
This result corresponds to the complex nullifier
\begin{equation}
	\left({\opvec p} - {\mat Z} {\opvec q}\right) \ket {\psi_{\mat Z}} = \vec 0\,,
\end{equation}
as expressed in Eq.~\eqref{eq:psiZnull},
with a complex adjacency matrix ${\mat Z}$
having imaginary diagonal entries $ie^{-2 r_k}$
and the remaining entries being either~0 or~1 depending
on the particular CV cluster state with unweighted edges.
For example, for two modes this reproduces the result
in Eq.~\eqref{eq:canontwomode} for a canonical two-mode CV cluster state.
More generally, the result corresponds to
complex-weighted graphs including self-loops.


Any $N$-mode Gaussian pure state can be built from
$N$ squeezed vacua sent through passive linear optics (modulo
phase-space displacements)~\cite{Braunstein2005}. In terms of stabilizers, this means
that, without loss of generality, the stabilizers of
$N$ momentum-squeezed states are transformed as
\begin{multline}
\label{eq:generalStab}
	\exp \left(-\tfrac 1 2 e^{-2 r_k} s^2 \right)\op U \op X_k(s)\op U^\dagger \op U \op Z_k(i e^{-2 r_k} s)\op U^\dagger \\
	= \exp \left[-is \left(\op p_k' - i e^{-2 r_k} \op q_k' \right) \right]\,,
\end{multline}
where~$\op p_k'$ and~$\op q_k'$ are the linearly transformed
momentum and position operators, respectively, 
after the corresponding (inverse) unitary transformation~$\op U$.
Provided this~$\op U$ represents a Gaussian transformation,
the exponent on the right-hand-side of Eq.~\eqref{eq:generalStab} will always be linear combinations of the generators~$\op q$ and~$\op p$. This would
include the canonical $\CZ$~interactions, as discussed before.
However, now we shall restrict ourselves to only passive, number-preserving unitaries~$\op U$,
without loss of generality~\cite{Braunstein2005}.
The canonical case would then require that the squeezing parts of the
$\CZ$ gates be absorbed into the offline momentum squeezers,
corresponding to Bloch-Messiah reduction~\cite{vanLoock2007}.

For the case of a passive linear transformation,
we can write $\opvec a' = \mat L \opvec a$, where $\mat L := \mat X + i\mat Y$~is an~$N \times N$ unitary matrix (with real~$\mat X$ and~$\mat Y$), and therefore [Cf.~the last matrix in Eq.~\eqref{eq:preiwasawa}],
\begin{align}
\label{eq:qppassive}
	\begin{pmatrix}
		\opvec q'	\\
		\opvec p'
	\end{pmatrix}
	=
	\begin{pmatrix}
		\mat X 	& -\mat Y	\\
		\mat Y 	& \mat X
	\end{pmatrix}
	\begin{pmatrix}
		\opvec q	\\
		\opvec p
	\end{pmatrix}
	\,.
\end{align}
Finally, through Eq.~\eqref{eq:generalStab}, we arrive at
the new stabilizers
\begin{multline}
\label{eq:generalStab_out}
	\exp \Bigl\{ -is
	 \sum_l \bigl[ ( Y_{kl} - i e^{-2 r_k} X_{kl} ) \op q_l \\ + ( X_{kl} + i e^{-2 r_k} Y_{kl} ) \op p_l \bigr] \Bigr\} \,.
\end{multline}
For the nullifiers, we obtain
\begin{equation}
	\left({\mat A} {\opvec p} - {\mat B} {\opvec q}\right) \ket {\psi_{\mat Z}} = \vec 0\,,
\end{equation}
where
\begin{align}
	A_{kl} &= X_{kl} + i e^{-2 r_k} Y_{kl}\,, \\
	-B_{kl} &= Y_{kl} - i e^{-2 r_k} X_{kl}\,,
\end{align}
which we may rewrite as
\begin{align}
	\left({\mat A}^{-1}{\mat A} {\opvec p} - {\mat A}^{-1}{\mat B} {\opvec q}\right) \ket {\psi_{\mat Z}}
	&= \left({\opvec p} - {\mat Z} {\opvec q}\right) \ket {\psi_{\mat Z}} = \vec 0\,,
\end{align}
with ${\mat Z}:= {\mat A}^{-1}{\mat B}$.  This gives us again the complex adjacency matrix
for an arbitrary pure Gaussian $N$-mode state.
We note that there are at most $4 N^2$ parameters
to determine the stabilizer/nullifier~\cite{Bartlett2002}.  These, however, are not independent,
as $\mat L$ must be a unitary matrix, and ${\mat B}$ follows from ${\mat A}$.
A general Gaussian unitary transformation has $2 N^2 + N$ free parameters,
without displacements, which is the same number for representing
a symplectic transformation from ${\rm Sp}(2N,\reals)$.
For representing pure Gaussian $N$-mode states (modulo
displacements), it is enough to apply a general Gaussian unitary transformation
to an $N$-mode vacuum state, where after Bloch-Messiah reduction~\cite{Braunstein2005},
the first passive transformation has no effect on the vacuum~\cite{vanLoock2007}.
Thus, $N$~real squeezing parameters~$r_k$ and $N^2$~parameters
for the remaining passive transformation~$\mat L$ suffice
to uniquely determine the matrices~${\mat A}$ and~${\mat B}$, and hence
the state through~${\mat Z}$.

As an example we refer to the standard two-mode squeezed state,
obtainable by interfering a $\op q$-squeezed state with a $\op p$-squeezed
state at a 50:50 beamsplitter~\cite{Furusawa1998,vanLoock1999,Braunstein2005a}. Then we have
\begin{align}
	\mat A &= \frac{1}{\sqrt{2}}
	\begin{pmatrix}
		1	&	1	\\
		i e^{-2 r_2}	&	-i e^{-2 r_2}
	\end{pmatrix}
	\,, \\
	\mat B &= \frac{1}{\sqrt{2}}
	\begin{pmatrix}
		i e^{-2 r_1}	&	i e^{-2 r_1}	\\
		-1	&	1
	\end{pmatrix}
	\,,\\
\intertext{and}
	\mat A^{-1} &= \frac{1}{\sqrt{2}}
	\begin{pmatrix}
		1	&	-i e^{+2 r_2}	\\
		1	&	i e^{+2 r_2}
	\end{pmatrix}
	\,,
\end{align}
resulting in the matrix ${\mat Z}_1$ from Eq.~\eqref{eq:TMSS_adjMatrix}
using equal initial squeezing $r_1=r_2=\alpha$.
Another simple example is the $N$-mode vacuum state,
for which $r_k = 0$, $\mat A = \mat L$, and $\mat B = i \mat L$,
so that ${\mat Z} = i \1$ for any $\mat L$. The vacuum always remains
an uncorrelated graph with only self-loops.

\vspace{1.5em}

\section{Mixed Gaussian states}
\label{app:mixed}

There is a very simple special case
of mixed Gaussian states
for which the entire pure-state graph calculus
presented in this article follows through as well
almost trivially.
This special case is usually referred to as
the $N$-dimensional isotropic oscillator~\cite{Simon1988}. %
The covariance matrix for a general mixed Gaussian $N$-mode
state is given by
\begin{align}
\mat \Sigma = \frac 1 2 \mat S \begin{pmatrix}
		\mat K	&	\vec 0	\\
		\vec 0	&	\mat K
	\end{pmatrix}
\mat S^\tp\,,
\end{align}
generalizing the expressions from Eq.~\eqref{eq:covfromS},
by replacing the $N$-mode vacuum/ground state $\1 /2$ by
$1/2$ times the above diagonal matrix with $\mat K = {\rm diag}(\mat\kappa)$,
where the vector $\mat\kappa = (\kappa_1,...,\kappa_N)^\tp$ contains
the symplectic eigenvalues (times 2) of $\mat \Sigma$. For pure states, we have
$\mat K = \1$.

Now let us assume that all symplectic eigenvalues are equal to $\kappa$,
corresponding to an $N$-mode Gaussian state
built from $N$ thermal states with identical excitation number,
replacing the $N$ initial vacua in Eq.~\eqref{eq:covfromS}.
Carrying along this one extra parameter $\kappa$, we can use our graphical formalism to describe this special case, as well.  We leave to future work a possible extension of our formalism to general Gaussian mixed states.


\bibliographystyle{bibstyleNCM}
\bibliography{CVgraph_trans}

\end{document}